\documentclass{l4dc2025}

\usepackage{upgreek}       

\usepackage[utf8]{inputenc} 
\usepackage[T1]{fontenc}    
\usepackage{hyperref}       
\usepackage{url}            
\usepackage{booktabs}       
\usepackage{amsfonts}       
\usepackage{nicefrac}       
\usepackage{microtype}      
\usepackage{xcolor}         
\usepackage{stmaryrd}
\usepackage{algorithm}
\usepackage{algorithmic}
\usepackage{multirow}
\usepackage{makecell}

\usepackage{array}
\usepackage{caption}
\usepackage{stackengine}
\usepackage{times}

\renewcommand{\em}{\it}

\newcommand{\RU}{R}
\newcommand{\RUinv}{\RU^{-1}}

\newcommand{\etabar}{{\eta}}

\newcommand{\ud}{\mathrm{d}}
\def\Re{\mathbb{R}}
\def\R{\mathbb{R}}
\def\argmin{\mathop{\text{\rm arg\,min}}}

\def\transpose{{\hbox{\rm\tiny T}}}

\newcommand{\E}[1]{\mathbb{E}\left[#1\right]}

\def\Expect{{\sf E}}

\newcommand{\cov}{\text{Cov}}
\newcommand{\backward}[1]{\overset{\shortleftarrow}{#1}}
\newcommand{\Ybarbar}{{Y}}
\newcommand{\SN}{S^{(N)}}

\newcommand{\trace}{\text{Tr}}

\newcommand{\OmegaN}{\Omega^{(N)}}

\newcounter{rmnum}

\newcounter{anum}

\newcommand{\Ricc}{\text{Ricc}}

\def\FRAC#1#2#3{\genfrac{}{}{}{#1}{#2}{#3}}
\def\half{{\mathchoice{\FRAC{1}{1}{2}}%
		{\FRAC{2}{1}{2}}%
		{\FRAC{3}{1}{2}}%
		{\FRAC{4}{1}{2}}}}

\newcommand{\inv}{^{-1}}
\newcommand{\tp}{^\transpose}
\newcommand{\normal}{\mathcal{N}}
\newcommand{\vf}{\mathcal{I}}
\newcommand{\correct}{\mathcal{C}}
\newcommand{\grad}{\nabla}
\newcommand{\pbart}{\bar{p}_t}

\newcommand{\bij}{\psi}

\newcommand{\stepsize}{\tau}
\newcommand{\iid}{\stackrel{\text{i.i.d}}{\sim}}
\newcommand{\identity}{\mathbb{I}}
\newcommand{\simulator}{\mathcal{S}}
\newcommand{\ip}[2]{\langle #1, #2 \rangle}

\newcommand{\sgn}{\mathrm{sgn}}
\newcommand{\ones}{\mathbf{1}}

\newtheorem{assumption}{Assumption}

\newcommand{\coloneqq}{:=}

\title{Interacting Particle Systems for Fast Linear Quadratic RL}

\author{%
 \Name{Anant A. Joshi} \Email{anantaj2@illinois.edu}\\
\addr University of Illinois Urbana-Champaign 
 \AND
 \Name{Heng-Sheng Chang} \Email{hschang2@illinois.edu}\\
 \addr University of Illinois Urbana-Champaign %
 \AND 
 \Name{Amirhossein Taghvaei} \Email{amirtag@uw.edu}\\
 \addr University of Washington Seattle 
 \AND
 \Name{Prashant G. Mehta} \Email{mehtapg@illinois.edu}\\
 \addr University of Illinois Urbana-Champaign 
 \AND
 \Name{Sean  P. Meyn} \Email{meyn@ece.ufl.edu}\\
 \addr University of Florida at Gainesville
}

\begin{document}

\maketitle

\begin{abstract}

This paper is concerned with the design of algorithms based on systems of interacting particles to represent, approximate, and learn the optimal control law for reinforcement learning (RL). The primary contribution is that convergence rates are greatly accelerated by the interactions between particles. Theory focuses on the linear quadratic stochastic optimal control problem for which a complete and novel theory is presented. Apart from the new algorithm, sample complexity bounds are obtained, and it is shown that the mean square error scales as $1/N$ where $N$ is the number of particles. The theoretical results and algorithms are illustrated with numerical experiments and comparisons with other recent approaches, where the faster convergence of the proposed algorithm is numerically demonstrated.

\end{abstract}


\section{Introduction}

This paper concerns approaches to reinforcement learning (RL) based on the construction of interacting particle systems. The development is in continuous time, and the state is assumed to evolve according to a linear stochastic differential equation (SDE),
\begin{equation}\label{eq:dyn}
\ud X_t = (AX_{t} + B U_t)\ud t + \sigma \ud W_t,\quad X_0=x
\end{equation}
where $X:=\{X_t:0\leq t\leq T\}$ is the $\Re^d$-valued state process, $U:=\{U_t:0\leq t\leq T\}$ is the $\Re^m$-valued control input, and $W:=\{W_t:0\leq t\leq T\}$ is a standard Brownian motion (B.M.), and
$A,B,\sigma$ are matrices of appropriate dimensions. 

The proposed approach is related to actor-only methods, also known as the policy optimization (PO) approach,  of which Williams' REINFORCE algorithm is most classical 
\cite{williams-1992}. 
The linear model is the subject of recent work in PO: two types of optimal control objectives have been considered, namely, 
linear quadratic Gaussian (LQG)~\citep{basei-2022} linear exponential quadratic Gaussian (LEQG)~\citep{zhang-2021-neurips,roulet-2020-acc}, 
and average cost versions of these~\citep{krauth-2019,
yadkori-2019,yang-2019-neurips,cassel-2021,
yaghmaie-2023,
hernandez-2023}.

A standard PO approach in the linear quadratic setting is special because the policies $ \{ \upkappa^\theta :  \theta\in\Re^n \}$ may be chosen deterministic and linear.  A basic algorithm is described as the following recursion: 
Starting from an initial stabilizing gain $K^0$, a sequence of gains $\{K^j:j=1,2,\hdots,M\}$ are learnt.  During the $j$-th iteration, the gain $K^j$ is evaluated by simulating $N$ copies of the model over a time-horizon:
\begin{subequations}
  \label{eq:dyn_PO}
  \begin{align}
    \ud X_t^i &= (AX_{t}^i + B K_t^j X_{t}^i )\ud t + \sigma \ud W_t^i,\quad 0\leq t\leq T,\quad 1\leq i\leq N\\
    X^i_0 &\stackrel{\text{i.i.d}}{\sim} \mathcal
N(0,I),\quad 1\leq i\leq N
    \end{align}
\end{subequations}
These evaluations are helpful to compute the gain $K^{j+1}$ through a gradient-descent procedure.     

When $N=1$ the algorithm might be applied using observed samples from a physical system;  otherwise, this technique requires a simulator to generate particles.    
The main message of this paper is that the use of a simulator combined with 
 carefully designed mean-field interactions between simulations (the particles) will ensure far greater efficiency in the learning process.  



\paragraph{Algorithm proposed in this paper.}
Simulate an interacting particle system:
\begin{subequations}\label{eq:dual_enkf_intro}
\begin{align}
\ud Y_t^i &= \underbrace{AY_{t}^i \ud t + B \ud \eta_t^i + \sigma \ud W_t^i}_{\text{copy of model}} + \underbrace{{\cal A}_t(Y_t^i; p_t^{(N)})
  \ud t}_{\text{mean-field interaction}},\quad 0\leq t\leq T,\quad 1\leq i\leq N \label{eq:dyn-N} \\
    {Y}^i_T  &\stackrel{\text{i.i.d}}{\sim} \mathcal
N(0,\mathcal{Y}),\quad 1\leq i\leq N, 
\end{align}
\end{subequations}
where $p_t^{(N)}$ is the empirical distribution of the ensemble $\{Y_t^i:1\leq i \leq N\}$. The specification of the terminal condition at time $t=T$ means that the system is simulated backward-in-time.  The three design variables are as follows:
\begin{itemize}
\item[(i)] $\mathcal{Y}$ is the covariance matrix to sample the $N$ particles at the terminal time. 
\item[(ii)] $\eta:=\{\eta^i:1\leq i\leq N\}$ where $\eta^i:= \{\eta_t^i, :0\leq t\leq T\}$ is the control input for the $i$-th particle. These inputs are designed to be independent B.M. with a prescribed covariance.   
\item[(iii)] ${\cal A}:= \{{\cal A}_t:0\leq t\leq T\}$ is a mean-field process which couples the simulations. The phrase ``mean-field'' means that the coupling depends {\em only} upon the (empirical) distribution $p_t^{(N)}$.
\end{itemize}
The triple $(\mathcal{Y},\eta,{\cal A})$ are designed with the goal that the empirical covariance of the ensemble $\{Y_t^i:1\leq i \leq N\}$ approximates the solution of the differential Riccati equation (DRE) at time $t$.  The resulting system is referred to as the {\em dual ensemble Kalman filter}. 




\paragraph{\textbf{Contributions:}} The paper builds on \cite{anant-2022} to include stochastic control systems.  The novel aspects of the present paper are three-fold:  \textbf{(i)} the algorithms and the analysis are extended to stochastic and robust/risk sensitive settings of the problem in a single unified framework; \textbf{(ii)} an error analysis is provided for finite but large $N$ interacting particle system; and \textbf{(iii)} sample complexity results are obtained and comparison of the same provided against state-of-the-art algorithms for linear quadratic RL (iv) comparisons are presented to the path integral control framework, which is  related to our approach and used widely in RL and robotics. (v) This paper includes numerical simulation and comparison for two benchmark example problems from previous works in this area.  

The salient features of the proposed algorithm are as follows: \textbf{(i)} It is not necessary that the matrix $A$ is Hurwitz or that a stabilizing gain matrix $K^0$ is known (this is an assumption in many of the prior studies on PO); and \textbf{(ii)} convergence theory relies on law of large numbers (LLN) and spectral constant known from the DRE theory.  Specifically, as $N\to\infty$, the proposed algorithm yields a learning rate that approximates the exponential rate of convergence of the solution of the DRE; and \textbf{(iii)} it alleviates issues like weight collapse and curse of dimensionality that are inherited by algorithms based on the importance sampling paradigm. 


\section{Problem formulation}
\textbf{Notation:}$\|\cdot\|_F$ denotes Frobenius norm for matrices,  $|\cdot|$ denotes 2-norm for vectors and $|\cdot|_M$ denotes weighted 2-norm under positive definite matrix $M$, that is, $|z|_M := z\tp M z$, $\normal(\text{mean},\text{covariance})$ denotes normal distribution, $\identity$ is used for identity matrix.

In linear quadratic settings, the cost function is quadratic as follows:
\[
c(x,a) = \half |C x|^2 + \half |a|_R^2,\quad x\in\Re^d,\;a\in\Re^m
\]

\begin{table}
\caption{Expressions for DRE, where $D:= B  \RUinv B^\transpose$ and $\Sigma := \sigma\sigma\tp$.}
\centering
\vspace{1em}
\begin{tabular}{@{}ccc@{}} \toprule
Cost & ${\mathcal{D}}(\Lambda)$ & ${\mathcal{D}^{\dagger}}(\Lambda)$  \\ \midrule
LQG & {$A^\transpose \Lambda + \Lambda A + C^\transpose C - \Lambda D \Lambda$}& {$A\Lambda + \Lambda A^\transpose - D + \Lambda C^\transpose C \Lambda$} \\ \midrule
LEQG & $ A^\transpose \Lambda + \Lambda A + C^\transpose C  - \Lambda(D - \theta \Sigma ) \Lambda$ & $ A\Lambda + \Lambda A^\transpose - \frac{1}{|\theta|}(D - \theta\Sigma) + |\theta|\Lambda C^\transpose C \Lambda$ \\ \bottomrule
\end{tabular}
\label{tb:dre}
\end{table}

Based on this, the following types of stochastic optimal control problems, linear quadratic Gaussian (LQG), linear exponential quadratic Gaussian (LEQG), and their average counterparts are considered (with $\theta \in \Re \setminus \{0\}$):
\begin{align}
J_T^{\text{\tiny LQG}}(U) &:= \E{ \int_0^T c(X_t,U_t) \ud t + \half\ |X_T|^2_{G}}, \tag{\text{LQG}} \\
J_T^{\text{\tiny LEQG}}(U) &:= \theta^{-1}\log\mathbb{E} \left[ \exp{\theta \left\{ \int_0^T c(X_t,U_t)  \ud t + \ \half|X_T|^2_{G}   \right\} } \right],  \tag{\text{LEQG}}\\
J^{\text{AVG,}i}(U) &:= \limsup_{T \to \infty}\frac{1}{T}J_T^{i}(U),\qquad i \in \{ \text{LQG}, \text{LEQG} \}. \tag{\text{AVG}}
\end{align}
For the LEQG problem, $\theta$ is referred to as the risk parameter: The case $\theta > 0$ is known as risk-averse and $\theta < 0$ as risk-seeking \cite{nagai-2013}. The problem is to choose the control $U$ to minimize the respective value $J(U)$ subject to the linear Gaussian dynamics~\eqref{eq:dyn}. A standard set of assumptions--that are also made here--are now listed.

\begin{assumption} \label{assn:model}
$(A,B)$ is controllable, and $C\tp C, R,\,G\succ 0$  and for LEQG, $BR\inv B\tp - \theta\sigma\sigma\tp \succ 0$. 
\end{assumption}

The main point of difference from the classical treatment is that the linear Gaussian model~\eqref{eq:dyn} is available {\em only} in the form of a simulator.

\begin{definition}[Simulator]
A simulator of \eqref{eq:dyn}, denoted $\simulator$, takes the current state $x\in\Re^d$, control $a\in\Re^m$ and (small) time-step $\stepsize$ as input  and gives the following random  variable as output
\[
  \simulator(x,a;\stepsize) = (Ax + Ba)\stepsize + \sigma\Delta W \quad \text{where}  \;\;\Delta W \iid \normal(0,\identity\stepsize).
\]
\label{defn:sim}
\end{definition}

\begin{remark}[Simulations and RL] 
\label{rmk:sim}
The simulator takes the current state and control at every call, and the random variables $\Delta W$ are i.i.d. from $\normal(0,\tau)$ for every simulator call.
A standard assumption in RL is that the state is available at every time $t$.  Outside of a simulation type setting, it is difficult to describe a system where such an assumption holds: Most real-world systems have partial observation of the states through noisy sensor outputs.  Next, many types of RL algorithms implement multiple iterations of the type~\eqref{eq:dyn_PO}, e.g., \cite[Algorithm 1,2]{krauth-2019}, \cite[Algorithm 2]{yang-2019-neurips}, \cite[Algorithm 1]{basei-2022}, \cite[Algorithm 1]{cassel-2021} \cite[Algorithm 3]{yaghmaie-2023}, \cite[Algorithm 2]{zhang-2021-neurips}, \cite[Algorithm 2]{cui-2023-l4dc}, \cite[Algorithm 1]{lai-2024}.  More discussion in Appendix \ref{app:theory}.
\end{remark}

\subsection{Riccati equation and the Q function}

Consider a matrix-valued process $\{P_t:0\leq t\leq T\}$ obtained from solving the DRE as follows:
\begin{align}
\label{eq:Riccati}
-\frac{\ud}{\ud t} P_t &= \mathcal{D}(P_t), \quad 0\leq t\leq T,\quad 
  P_T = G &
\end{align}
where the expressions for the Riccati operator $\mathcal{D}(\cdot)$ are given in Table \ref{tb:dre}.  While the DRE is the optimality equation for the finite time-horizon, the average cost solution is obtained by letting the time-horizon $T\to\infty$.  Because $(A,B)$ is controllable and $(A,C)$ is observable, for any fixed time $t$, $P_t \to \bar{P}$ which solves the ARE: $\mathcal{D}(\bar{P}) = 0$ (\cite[Theorem 3.7]{ksivan}). 

\begin{definition}[Q-function] The continuous-time Q-function (or Hamiltonian) is defined as
\begin{align*}
  \mathcal{Q}(x,a;t) & := c(x,a)   + x{\tp} P_t (Ax + Ba),\quad 0\leq t\leq T, \;\;x\in\Re^d,\;a\in\Re^m \quad &{\text{for LQG, LEQG}}\\
  \bar{\mathcal{Q}}(x,a) & := c(x,a)   + x{\tp} \bar{P} (Ax + Ba),\quad x\in\Re^d,\;a\in\Re^m \quad &{\text{for AVG}}
\end{align*}
\end{definition}

Then (see~\cite{liberzon}),
\[
U_t^{\text{opt}} = \begin{cases} \argmin_{a\in\Re^m} \mathcal{Q}(X_t,a;t),\quad 0\leq t\leq T, \quad &\text{LQG, LEQG}\\[10pt]
\argmin_{a\in\Re^m} \bar{\mathcal{Q}}(X_t,a),  \quad &\text{AVG}
\end{cases}
\]
Because the Q function is quadratic, it is easily verified that the optimal control law is linear:
\[
U_t^{\text{opt}} = \begin{cases} K_t X_t, \quad K_t := -R^{-1}B\tp P_t, \quad 0\leq t\leq T, \quad &\text{LQG, LEQG}\\
 \bar{K} X_t , \quad \bar{K} := -R^{-1}B\tp \bar{P}, \quad &\text{AVG}
  \end{cases}
\]

The analysis of this paper requires consideration of $P_t\inv$.  
Since $G \succ 0$, it holds that $P_t\succ 0$ for $0\leq t\le T$~\cite[Sec.~24]{brockett2015finite}.  Therefore, $P_t\inv$ is well-defined. For every $t \in [0,T]$,
\begin{align} \label{eq:St}
S_t := P_t\inv \text{ for LQG;} \qquad \text{and} \qquad
S_t := (|\theta|P_t)^{-1} \text{ for LEQG}.
\end{align}
Then $\{S_t:0\leq t\leq T\}$ solves the dual DRE as follows:
\begin{align*}
  -\frac{\ud}{\ud t} S_t &= \mathcal{D}^{\dagger}(S_t), \quad 0\leq t\leq T,\quad 
  S_T  = G^{-1}
\end{align*}

\section{Interacting particle algorithm}
\label{sec:alg}
In this section, two sets of algorithms are described to approximate the optimal control law based only on the use of the simulator.  These are as follows:

$\bullet$ \textbf{Offline algorithm for solving DRE.} The goal is to learn an approximation of the Q-function.  These approximations for the finite time-horizon and the average cost problems are denoted as $\mathcal{Q}^{(N)}$ and $\bar{\mathcal{Q}}^{(N)}$, respectively. Additional background and justification appears in Appendix \ref{app:theory}.


$\bullet$ \textbf{Online algorithm for computing the optimal control.}  For each fixed time $t$, the optimal control is obtained by taking an arg min of approximate Q-function.

\subsection{Dual EnKF for approximating solution of DRE}

Consider the interacting particle system~\eqref{eq:dual_enkf_intro}.  The triple $(\mathcal{Y}, \eta, {\cal A})$ is designed as follows:

\textbf{(i) Design of $Y_T^i$:} Sample $Y_T^i \iid \normal(0,S_T)$ for $i=1,2,\hdots,N$.

\textbf{(ii) Design of $\eta^i$:} The input $\eta^i$ are i.i.d copies of a B.M. $\eta$ whose covariance is
\begin{align}
\cov(\eta) = 
R\inv\text{ for LQG}, \quad \text{ and } \quad 
\cov(\eta) =  (\sqrt{|\theta|}R)\inv \text{ for LEQG}.
\label{eq:cov}
\end{align}
In the context of RL, $\eta$ has an interpretation as the exploration signal.  The form~\eqref{eq:cov} of the covariance means that the cheaper control directions are explored more.

\textbf{(iii) Design of $\mathcal{A}_t$:} The interaction term is a mean-field type linear control law as follows:
\begin{align}
\mathcal{A}_t(z;p_t^{(N)}) \coloneqq 
\begin{cases}
\half L_t^{(N)}C(z+n_T^{(N)}) + \half \Sigma (S_t^{(N)})\inv(z - n_t^{(N)}); \quad &\text{LQG} \\ \noalign{\vskip9pt}
\frac{|\theta|}{2} L_t^{(N)}C(z+n_T^{(N)}) + \sgn(\theta)\Sigma (S_t^{(N)})\inv(z - n_t^{(N)}); \quad &\text{LEQG}
\end{cases}
\label{eq:At}
\end{align}
where $\Sigma: =\sigma \sigma^\transpose$, $n_t^{(N)}:=N^{-1} \sum_{i}{Y_t}^i$,  and 
\begin{align*}
L_t^{(N)} &:= \tfrac{1}{N-1}\sum_{i=1}^N
  ({Y_t}^i - n_t^{(N)} )(C{Y_t}^i - C n_t^{(N)} )\tp,  \,  \SN_t := \tfrac{1}{N-1}\sum_{i=1}^N({Y}^i_t-n^{(N)}_t)({Y}^i_t-n^{(N)}_t)\tp.
\end{align*}

From~\eqref{eq:St}, provided the right-hand side is well-defined,
\begin{align} \label{eq:St_inverse}
P_t^{(N)} := (S_t^{(N)})\inv\text{ for LQG,} \quad \text { and } \quad 
P_t^{(N)} := (|\theta|S_t^{(N)})^{-1}; \text{ for LEQG},
\end{align}
and for the average cost problem, $\bar{P}^{(N)}:=P_0^{(N)}$.
The error analysis is the subject of the following main result of this paper.

\begin{theorem}
 Consider the dual EnKF~\eqref{eq:dual_enkf_intro} under Assumption \ref{assn:model}.  Then for $N \ge d+1$, for each fixed $t$,
\begin{subequations} \label{eq:error-S}
\begin{align}
\text{(Finite-horizon)} \quad \Expect[\|S^{(N)}_{t}-{S}_{t}\|_F^2] &\leq
\frac{C_1}{{N}}, \quad \Expect[\|P^{(N)}_{t}-{P}_{t}\|_F^2] \leq \frac{C_4}{{N}}, \qquad 0\leq t\leq T, \label{eq:error-ST} \\
\text{(Average cost)} \quad \Expect[\|S^{(N)}_{t}-\bar{S}\|_F^2] &\leq
\frac{C_2}{{N}}, \quad \Expect[\|P^{(N)}_{t}-\bar{P}\|_F^2] \leq
\frac{C_5}{{N}}, \qquad
 \text{as}\;T\to\infty
\end{align}
\end{subequations}
(where $C_1,C_2,C_3,C_4$ are model dependent but time-independent constants). For the average cost problem, there exists a constant $\lambda>0$ such that exponential convergence to the stationary solution is obtained as follows:
\begin{align}
\Expect[\|S^{(N)}_{t}-\bar{S}\|_F^2] &\leq
\frac{C_2}{{N}} + C_3e^{-2\lambda (T-t)} \Expect[\|\SN_T-\bar{S}\|_F^2],  \qquad 0\leq t\leq T \label{eq:error-Sinf} 
\end{align}
\end{theorem}
\begin{proof}
These bounds are based on ~\cite{delmoral-2019-matrix-ricc}. The proof appears in Appendix \ref{app:delmoral}.
\end{proof}

Formula~\eqref{eq:error-Sinf} is important because $\lambda$ is the rate for learning the optimal solution.  The constant $\lambda$ is the spectral constant related to the exponential convergence of the solution of the DRE to the solution of the ARE \citep{ksivan}. As seen in the proof, the interaction term $\mathcal{A}$ is responsible for this propoerty. The formula is useful to see the relation between the simulation horizon $T$ and the error. For $\varepsilon > 0$, let $t=0$ in~\eqref{eq:error-Sinf}, $N>O(\frac{1}{\varepsilon^2})$ and $T>O(\log(\frac{1}{\varepsilon}))$, then error is smaller than $\varepsilon$.
The offline dual EnKF algorithm (Algorithm~\ref{alg:P} in Appendix \ref{app:alg}) presents a method to simulate the interacting particle system \eqref{eq:dual_enkf_intro} using only access to a simulator. 
For the numerical approximation of the SDE, a first order Euler-Maruyama method is used and may be replaced with a higher order method.

\subsection{Algorithm for approximating optimal control}

If the matrix $B$ is available, then the optimal control input at time $t$ is approximated as follows:
\[
U_t^{(N)} = \begin{cases} K_t^{(N)} X_t, \quad K_t^{(N)} := -R^{-1}B\tp P_t^{(N)}, \quad 0\leq t\leq T, \quad &\text{LQG, LEQG}\\
 \bar{K}^{(N)} X_t , \quad \bar{K}^{(N)} := -R^{-1}B\tp \bar{P}^{(N)}, \quad &\text{AVG}
  \end{cases}
\]

For the case where an explicit form of $B$ is not known, then the simulator is used to obtain an empirical approximation of the Q-function as follows:
\begin{definition}[Empirical Q-function] The empirical approximations are defined as
\begin{align}
  \mathcal{Q}^{(N)}(x,a;t,\stepsize) & := c(x,a)\stepsize   + x{\tp} P_t^{(N)}  \simulator(x,a;\stepsize),\quad 0\leq t\leq T, \;\;x\in\Re^d,\;a\in\Re^m \label{eq:empiricalQ}\\
  \bar{\mathcal{Q}}^{(N)}(x,a;\stepsize) & := c(x,a)\stepsize   + x{\tp} \bar{P}^{(N)} \simulator(x,a;\stepsize),\quad x\in\Re^d,\;a\in\Re^m \notag
\end{align}
\label{defn:q}
\end{definition}
Based on the empirical Q-function, the optimal control is given by
\[
U_t^{(N)}\stepsize = \begin{cases} \argmin_{a\in\Re^m} \E{\mathcal{Q}^{(N)}(X_t,a;t,\stepsize)|X_t},\quad 0\leq t\leq T, \quad &\text{LQG, LEQG}\\[10pt]
\argmin_{a\in\Re^m} \E{\bar{\mathcal{Q}}(X_t,a;\stepsize)|X_t},  \quad &\text{AVG}
\end{cases}
\]
The expectation on the right-hand side is necessary because the simulator is noisy.  A most straightforward implementation is to simply replace the expectation with a single sample---as one does in a stochastic gradient descent procedure.  With additional computational budget, the expectation is approximated through $N_e$ evaluations in a batch.

To evaluate the arg min, one may use a zero order optimization framework~\citep{bach-2016}.  A simpler algorithm is obtained by noting that, like the $Q$ function, the empirical $Q$ function is also a quadratic function of the state, of the form $\half a\tp R a + B\tp a + \varphi(x)$ where $\varphi(\cdot)$ is now a random function.  For the case when the number of control inputs $m$ is small, optimal control is approximated by evaluating the Q function for $a= R\inv e_i$ where $\{e_1,e_2,\hdots,e_m\}$ are basis vectors in $\Re^m$. 
Details of the procedure appear in Appendix \ref{app:errorU} where the resulting empirical approximation of the optimal gain is described and the following bound is shown:
\begin{align}
\Expect[\|\hat{K}^{(N)}_{t} -K_{t}\|_F^2] &\leq
\frac{C_6}{{N}} + \frac{n C_7}{{N_e\stepsize}},\quad 0\leq t\leq T \label{eq:error-U}
\end{align}
The online approximation of optimal control input is tabulated as Algorithm \ref{alg:EnKF} in Appendix \ref{app:alg} and an error analysis of the gain appears in Appendix \ref{app:errorU}.

\begin{remark}
An important observation needs to be made here. The result seems very surprising since $\tau$ appears in the denominator, which suggests that choosing a large simulation step size would result is better accuracy. This is happening because of a $\frac{W_\tau}{\tau}$ type term in the error analysis. What is hidden here, is that the error resulting from the discretization of the SDE has not been taken into account in the analysis, and it will yield a $O(\tau)$ type term in the error analysis. Thus the $O(\frac{1}{\tau}) + O(\tau)$ terms will lead to an optimal step size to choose.
\end{remark}

\subsection{Comparison of sample complexity to related works}

There are two types of errors for which analysis has been reported in recent literature: (i) the error in approximating the optimal value function; and (ii) the error in approximating the optimal gain matrix.  Most of these results are for the stationary average cost case in the stochastic setting of the problem or for the infinite-horizon linear quadratic regulator (LQR) in the deterministic ($\sigma=0$) setting. The quantitative comparisons with prior work are tabulated in Table \ref{tb:samp-comp}.  

\begin{table}[t]
\centering
\caption{Complexity bounds in terms of error $\varepsilon$.  These estimates are reported for the error in approximating gain in \cite[Theorem 4.3]{zhang-2021-neurips} and \cite[Theorem 2.2]{krauth-2019}; and the error in approximating the optimal cost in \cite[Lemma 6]{cassel-2021} and \cite[Theorem 4.3]{yang-2019-neurips}.}
\vspace{1em}
\begin{tabular}{cccc} \toprule
Algorithm & particles/samples & simulation time & iterations \\ \midrule
dual EnKF & $O(1/\varepsilon^2)$ & ${O}(1/\log(\varepsilon))$ & 1 \\ \midrule
\cite{zhang-2021-neurips} & $\tilde{O}(1/\varepsilon^4)$ & $O(1)$ & $O(1/\varepsilon)$ \\ \midrule
\cite{cassel-2021} & $\tilde{O}(1/\varepsilon^4)$ & $O(1)$ & $O(1/\varepsilon)$ \\ \midrule
\cite{krauth-2019} & 1 & $O(1/\varepsilon^2)$ & ${O}(1/\log(\varepsilon))$ \\ \midrule
\cite{yang-2019-neurips} & 1 & $O(1/\varepsilon^5)$ & ${O}(\log(1/\varepsilon))$ \\ \bottomrule
\end{tabular}
\label{tb:samp-comp}
\end{table}

In~\cite{krauth-2019}, an off policy method is used to estimate the Q function for discrete time average cost LQG. A linear function approximation is used with quadratic basis functions. The system is run for some fixed time using an exploration policy. At the end of each episode, the Q function is estimated using least squares. The error bounds in approximating the optimal gain are reported in~\cite[Theorem 2.2]{krauth-2019}. These results are closest to our work in terms of sample complexity requiring $O(\log(1/\varepsilon))$ training episodes and $O(1/\varepsilon^2)$ simulation time for error of $\varepsilon$ (see \cite[Theorem 2.2]{krauth-2019}). 

In \cite{yang-2019-neurips}, a policy gradient algorithm is described.  The actor is a gradient descent over the space of gains, where the policy gradient theorem is used to obtain the gradient. Error bounds are obtained for the error in value function \cite[Theorem 4.3]{yang-2019-neurips} which is related to error in solution of Riccati equation~\cite[Theorem 4.3]{yang-2019-neurips}.  The algorithm needs $O(\log(1/\varepsilon))$ iterations, and a simulation horizon of the order $O(1/\varepsilon)$ for an $\varepsilon$ error from the optimal value~\cite[Theorem 4.3]{yang-2019-neurips}.

In \cite{cassel-2021}, a zero order policy gradient algorithm is given for regret minimization in discrete time LQG. The idea is to perturb the gain in random directions to estimate the gradient of the value function with respect to the gain. Based on~\cite[Lemma 6]{cassel-2021}, $\tilde{O}(1/\varepsilon^4)$ samples are needed for gradient estimation and $O(1/\varepsilon)$ gradient descent iterations are needed for $\varepsilon$ error in approximating the optimal value.

On the LEQG problem, \cite{zhang-2021-neurips} extends the previous work of \cite{zhang-2020-l4dc}, \cite{zhang-2021-arxiv}, and studies model free policy gradient methods for finite-horizon discrete-time LEQG. The work utilizes the equivalence between LEQG and linear quadratic min-max game to describe a ``double-loop scheme''.  The approach is to write the optimization on the space of gains, and then apply a zeroth order policy optimization method to approximate the gradient flow.  A sample complexity analysis is given that quantifies the error bounds based on number of iterations and number of samples needed.  The algorithm requires $\tilde{O}(1/\varepsilon^4)$ samples to estimate the gradient, and $O(1/\varepsilon)$ of gradient descent iterations for $\varepsilon$ error in gain~\cite[Theorem 4.3]{zhang-2021-neurips}.

The trade-off between EnKF and policy gradient type or least-squares type algorithms is as follows. The latter class of methods typically require multiple iterations (episodes) for simulating a system over a finite time-horizon, albeit with a relatively smaller number of particles, while EnKF needs only a single iteration but with a larger number of particles. Notably, the work in \cite{yang-2019-neurips,krauth-2019,lai-2024,yaghmaie-2023} needs only a single copy of the system. The EnKF particles are simulated in parallel, giving rise to much more efficient and faster implementation. Moreover, EnKF does not require an initial feasible (stabilizing) gain, while \cite{krauth-2019}, \cite{yang-2019-neurips}, \cite{yaghmaie-2023}, \cite{cassel-2021}, \cite{zhang-2021-neurips}, \cite{cui-2023-l4dc}, \cite{lai-2024} need one.

\subsection{Conceptual comparison to path integral control}

While the focus of this paper is on designing the interaction term $\mathcal{A}$ in \eqref{eq:dyn-N} for the purpose of learning the value function, a related idea is the path integral approach(
 \cite{kappen-2005},
\cite{theodorou-2012},
\cite{kappen-2015}, 
\cite{williams-2016},
\cite{williams-2018}) 
In this class of algorithms, one works with the LQG problem under the assumption that $\Sigma = \lambda D$ for some $\lambda > 0$ (for simplicity we present formulas with $b = \sigma$ and $R = \identity$). One simulates multiple trajectories $\ud X_t^i = (AX_t^i + BU_t^i)\ud t + \ud W_t^i$, and approximates the value function as
\begin{align*}
\exp(-V(x,t)) &= \E{\exp\left(\int_s^T |CX_t|^2 \ud t -g^U(s,x)\right)} \approx \frac{1}{N}\sum_{i=1}^{N}w_i \\
U_s^*(x) - U_s &\approx \frac{\sum_{i=1}^{N} w_i W^{(i)}_{\tau}}{\tau \sum_{i=1}^{N} w_i}, \quad 
w_i := \exp\left(\int_s^T |CX^i_t|^2 \ud t-g^{U^i}(s,x)\right)
\end{align*}
%
%

with $g^U(s,T) := \int_s^T \half U_t \tp U_t dt  +U_t \tp dW_t$. 
%
%
%
%
The following are key differences between the two approaches: \textbf{(i)} In our approach, all particles have equal weight $1/N$. However, importance sampling is well known to suffer from particle degeneracy. The issue becomes severe in higher dimensions and is known as curse of dimensionality, which our algorithm avoids (see \cite{amir-cips} for a theoretical comparison).
\textbf{(ii)}
Another major point of distinction with is that the path integral control is a fully model-based algorithm  (formula (25) (where implementation of those terms assumes knowldge of $B$ and $\sigma$) and Section IV-B in \cite{williams-2016} and Section VI-B in \cite{williams-2018})) while we focus on a simulator-based setting through the design of interactions. Moreover, MPPI assumes a special relation between control cost and noise which we do not need (see for example, the equation above (24) in \cite{williams-2018} or \cite[equation (6)]{williams-2016}). \textbf{(iii)} A third major point is between how trajctories are utilised. Since MPPI needs to iteratively evaluate expectation under the controlled measure, they need to run multiple copies of the system in for each iteration, while we need to run them only once. Our algorithm has interaction between particles, while MPPI uses importance sampling based approaches. 

\begin{figure}[t]
		\centering
{
	\subfigure{
         \includegraphics[scale=0.2]{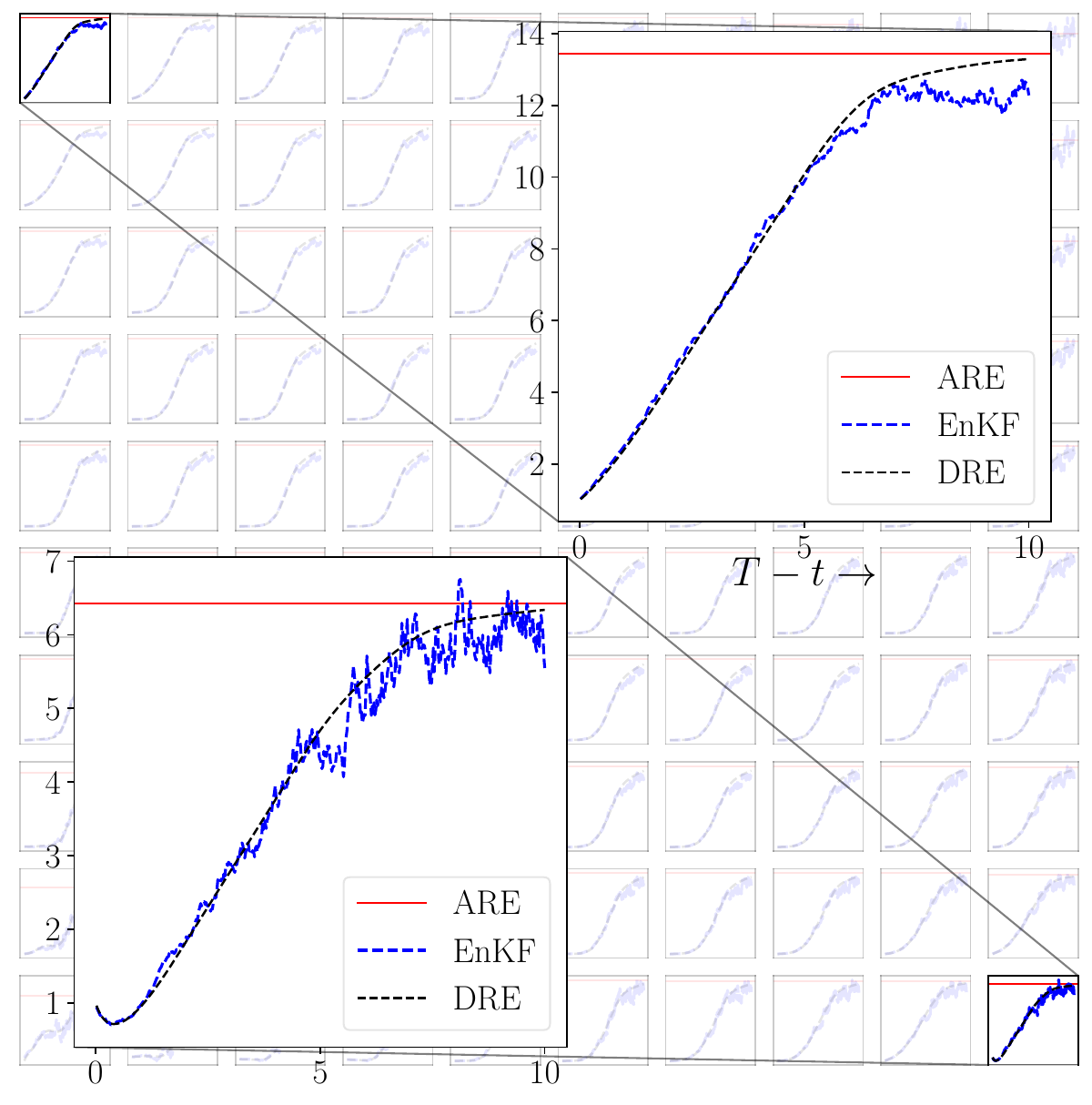}  
         \label{fig:LQG-learns}
	}
	\subfigure{
         \includegraphics[scale=0.2]{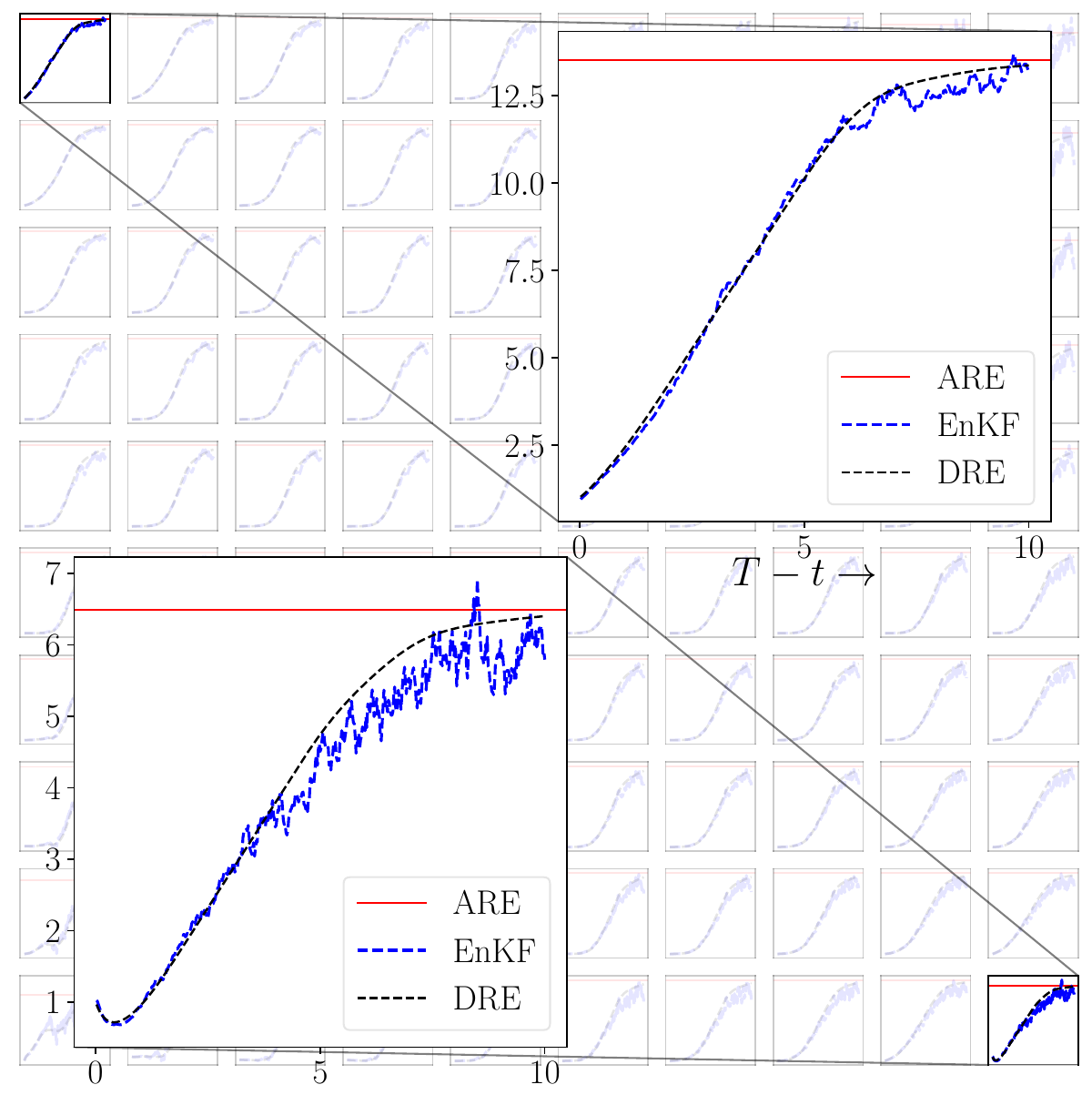}  
         \label{fig:LEQGP-learns}
     }
     \subfigure{
         \includegraphics[scale=0.2]{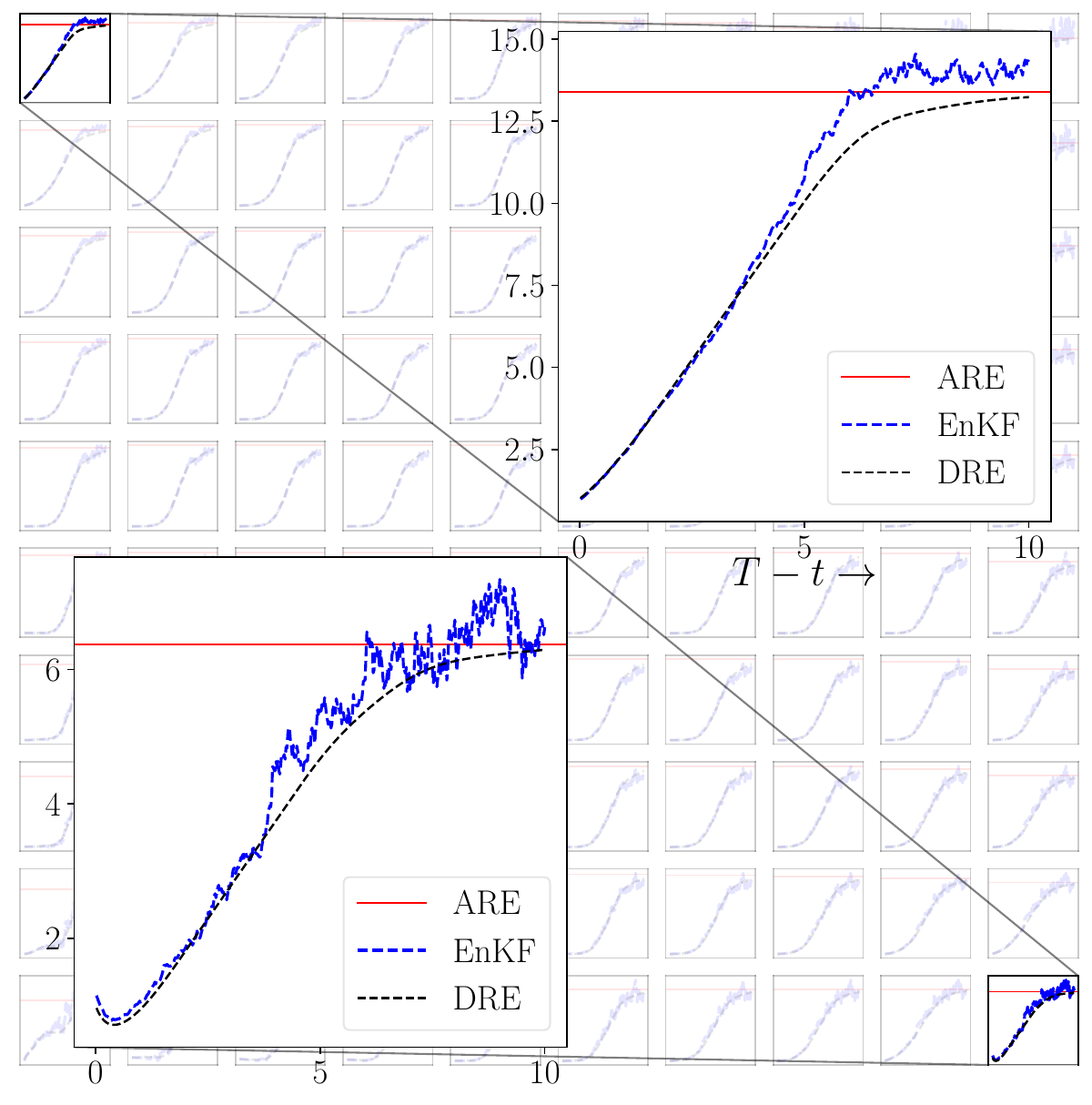}  
         \label{fig:LEQGN-learns}
     }
}
		\caption{Comparison of the numerical solutions obtained from the EnKF, the DRE, and the ARE. 
		The plots are in order: (a) LQG, (b) LEQG ($\theta > 0$) (c) LEQG ($\theta < 0$).}
		\label{fig:convergence}
		\vspace{-0.1in}
	\end{figure}

\begin{figure}
\centering
\includegraphics[scale=0.3]{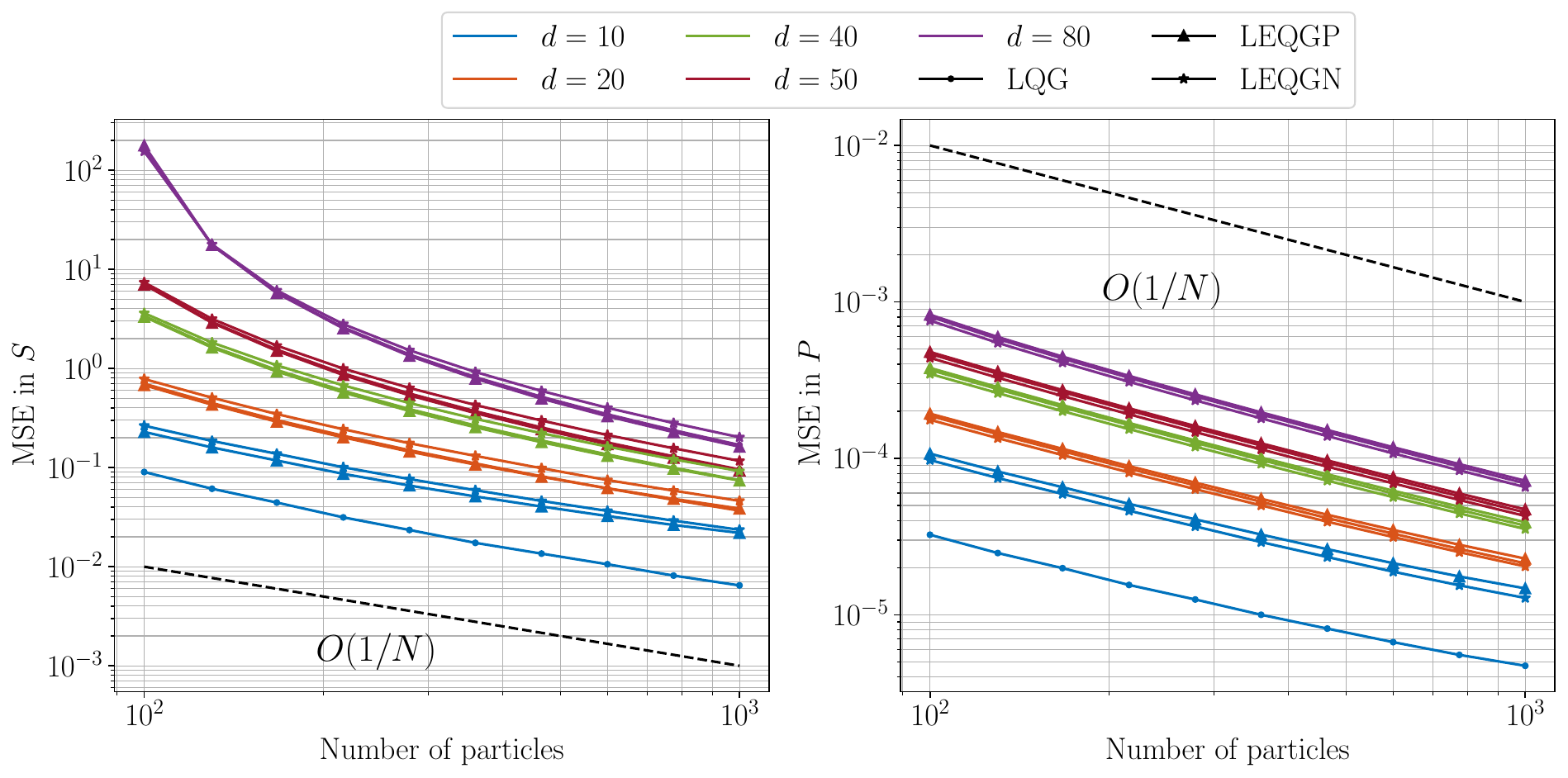}
\caption{Relative error in approximating the solution of the ARE by the dual EnKF.}
\label{fig:mse}
\end{figure}

\section{Numerical experiments and comparisons}
\label{sec:sim}


\subsection{Numerical illustration of error formulas~\eqref{eq:error-S},\eqref{eq:error-Sinf})}

An attractive feature of dual EnKF is that with large $N$, learning rate is inherited from the DRE convergence theory (see formula~\eqref{eq:error-Sinf}).  A numerical illustration of this formula, showing convergence of the $d^2$ entries of the $P$ matrix, is depicted in Figure ~\ref{fig:convergence}.  The model is $d=10$ dimensional where the entries of the $A$ matrix are randomly sampled (see Appendix \ref{app:addnumrandom} for details).  Five of the total ten eigenvalues of $A$ have positive real parts for the particular realization used in generating Figure~\ref{fig:convergence}.

In order to investigate scaling with increasing state dimension $d$, a spring mass damper model was introduced in \cite{mohammadi_global_2019}. For this model, all three controllers are evaluated (LQG and LEQG for $\theta$ positive and negative). 
The model and simulation parameters described in the Appendix \ref{app:smd}.
Figure~\ref{fig:mse} depicts the scaling as a function of $N$ for the following metrics:
$
 \frac{\Expect[\|\bar{S}^{(N)}-\bar{S}\|_F^2]}{\|\bar{S}\|_F^2} \text{ and } \frac{\Expect[\|\bar{P}^{(N)}-\bar{P}\|_F^2]}{\|\bar{P}\|_F^2}
 $
Consistent with \eqref{eq:error-S}, both the errors go down as $\frac{1}{N}$. Additional results on the performance of the optimal control law appear in the Appendix \ref{app:smd}.

\subsection{Numerical comparisons with prior work} 
\label{sec:comp}
\textbf{Policy optimization:}
For this study, a three dimensional discrete-time system from \cite{zhang-2021-neurips} is considered.  For this model, comparisons are made with the following: (i) Finite-horizon LEQG in \cite{zhang-2021-neurips}, denoted [Z21]; and (ii) Average cost LQG in \cite{krauth-2019}, denoted [K19].  
Comparison is done for  relative error in approximation of the optimal gain (with respect to the optimal gain)  and relative error in the cost incurred by the control (with respect to the optimal cost)  given by the algorithm.
Figure \ref{fig:comparison} depicts the numerically computed relationship between the relative error and the computational time.
 
\textbf{Path integral control:} We compare relative error in the cost incurred by the control (with respect to the optimal cost) given by the path integral approach \citep{williams-2018}, on the spring mass damper system for various dimensions of the state in Figure \ref{fig:mppi}.

 For each algorithm, the error becomes smaller with increasing computational time.  For the dual EnKF, this tradeoff is obtained by increasing the number of particles.  For [Z21] and [K19], the tradeoff is obtained by increasing the number of iterations and the time horizon. 
We observe that EnKF needs simulation times which are at least an order of magnitude lower than the other algorithms. 
See Appendix \ref{app:numcomp} for additional information on the optimal control problem and the simulation parameters, and additional discussion on these studies.

\begin{figure}[t]
  \centering
{
  \subfigure{
    \includegraphics[scale=0.3]{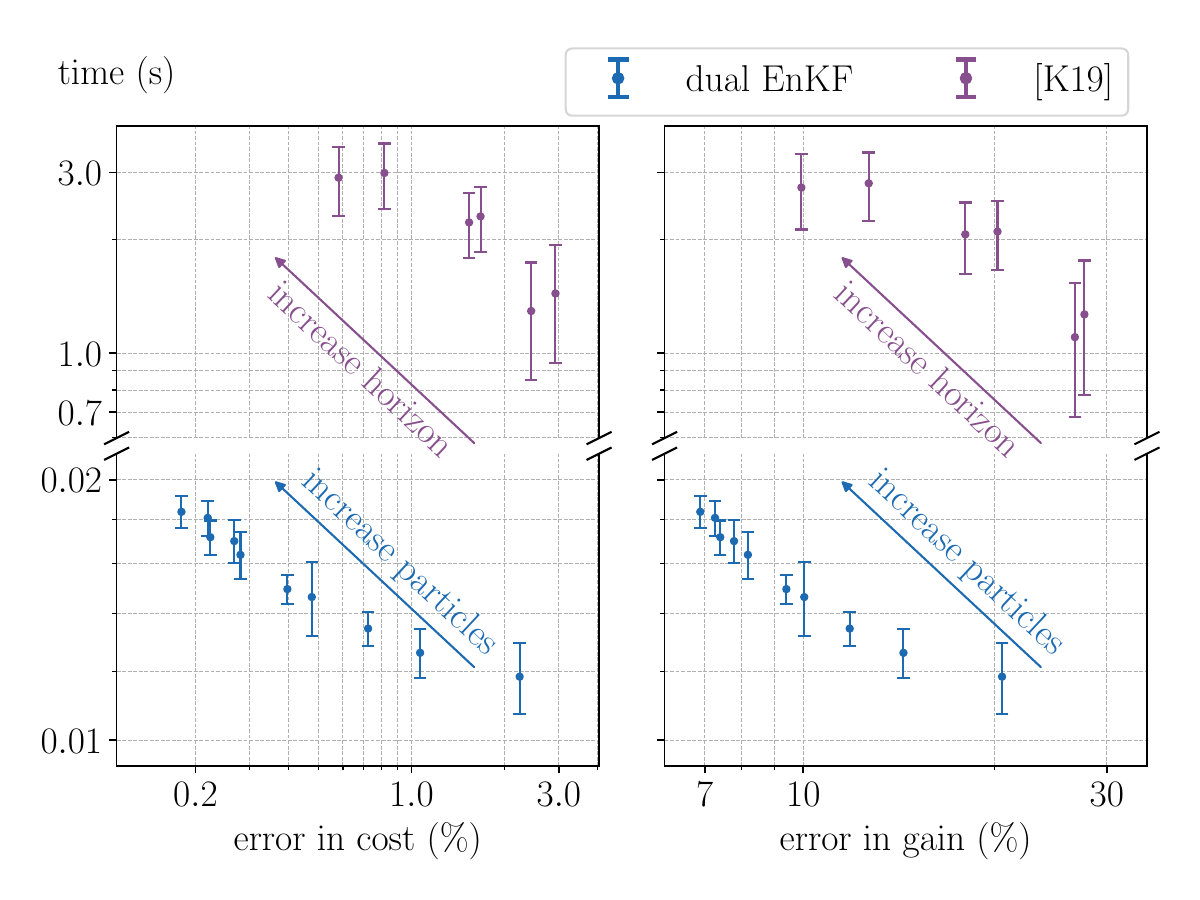}
    \label{fig:lqg_comparison}
  }
  \subfigure{
    \includegraphics[scale=0.3]{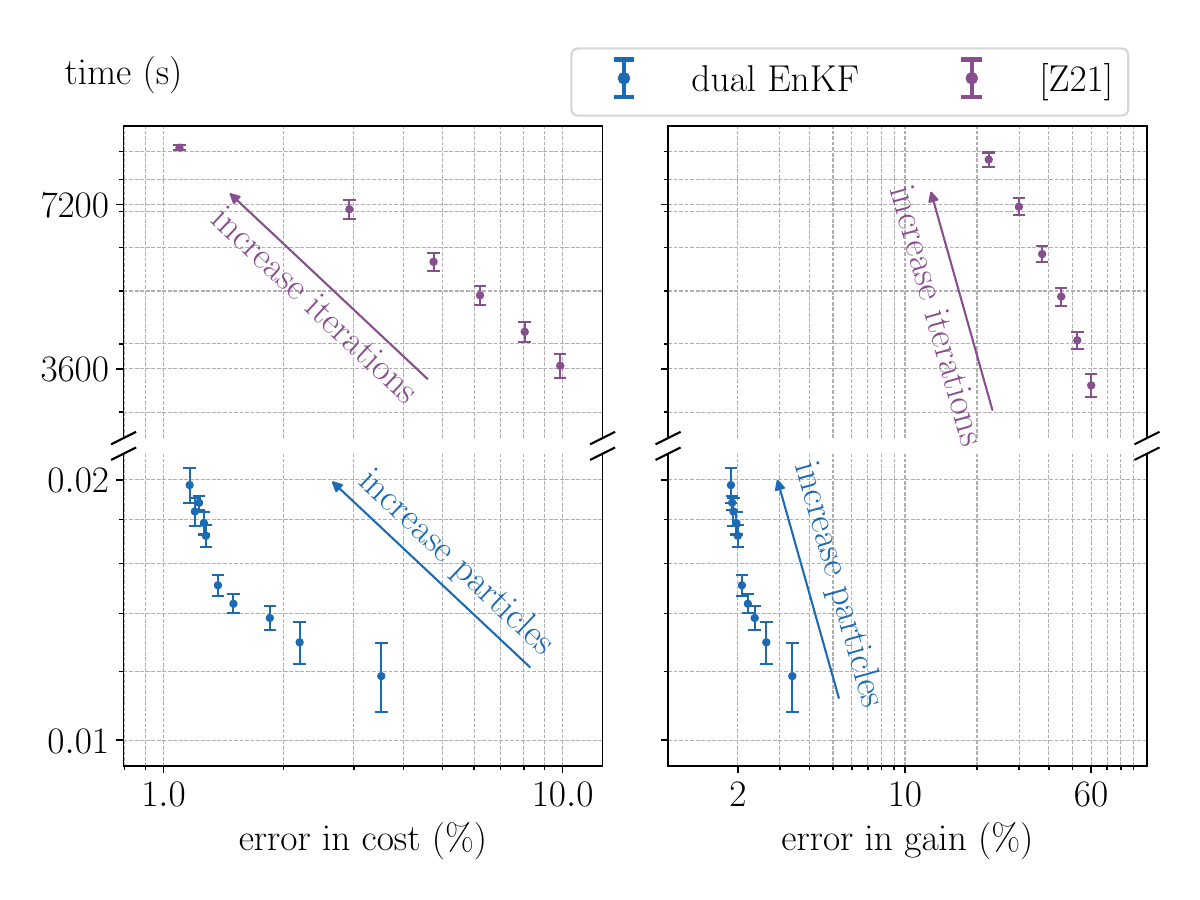}
    \label{fig:leqg_comparison}
  }
}
  \caption{
Comparison of dual EnKF with : (a) [K19] for  infinite horizon LQG; and (b) [Z21] for  finite horizon LEQG. See Section \ref{sec:comp} for details.
  }
  \label{fig:comparison}
\end{figure}

\begin{figure}
\centering
\includegraphics[scale=0.35]{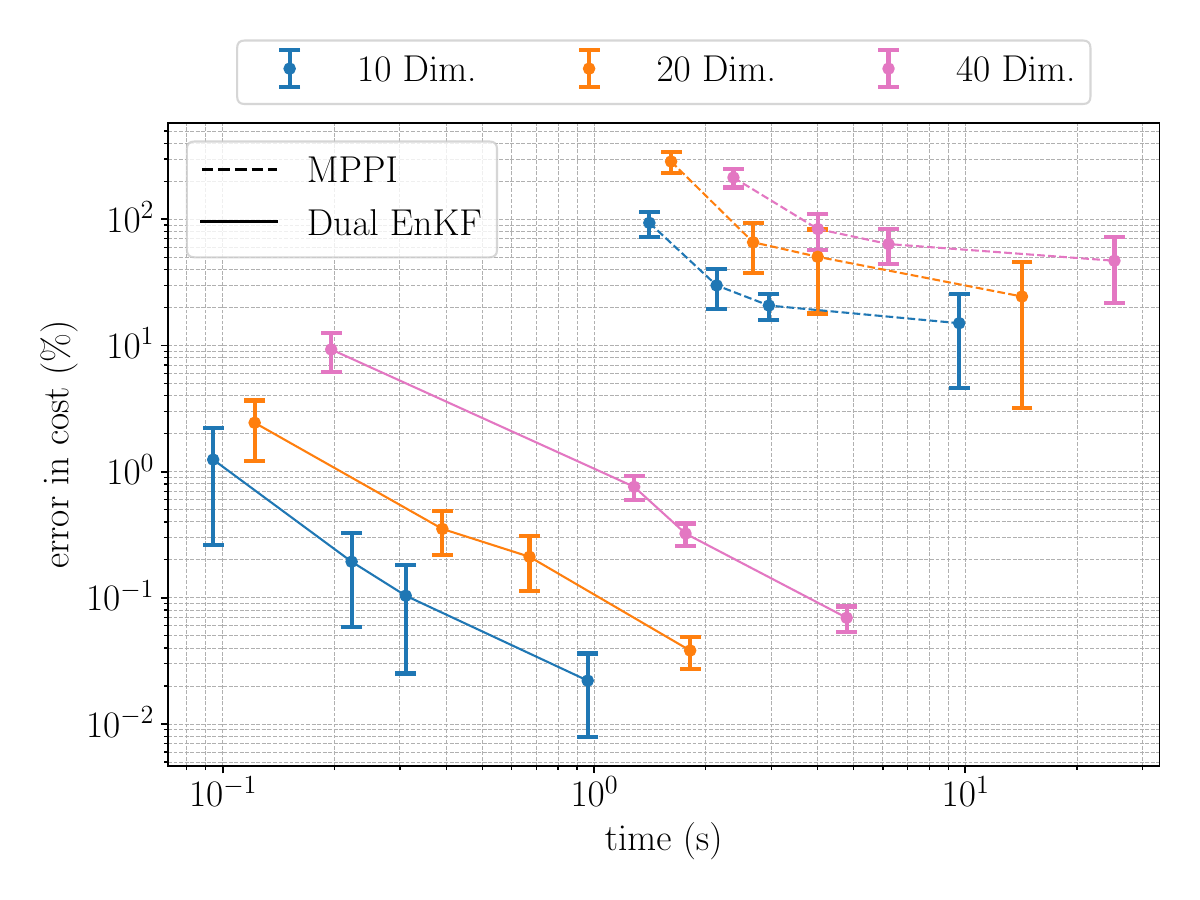}
\caption{Comparison of dual EnKF with path integral control for spring mass damper system.}
\label{fig:mppi}
\end{figure}

\newpage
\bibliography{refs,litsur,literature}

\appendix
\newpage

\section{Theory}
\label{app:theory}

\subsection{Simulator}

We assume access to a simulator but not explicit knowledge of the system matrices $A,B$. The simulator-based setting is  widely applicable in RL, for instance, in the ATARI games, open AI gym, MuJoCo and other such environments. In these settings, the transition probabilites are not explicitly known but one is allowed to sample the trajectory by running a simulation. In many scenarios, e.g., tic-tac-toe or chess or the game of go, transition probabilities are exactly known (these are used to construct the game). However, it is a common RL practice to approximately solve the dynamic programming equation only using simulations (having the computer play the game). We are doing something similar for LQ optimal control problem.

For the linear quadratic problems, previously published works by various research groups have also focused on the simulation-based setting (Remark \ref{rmk:sim}). The distinction and the advantage of our approach is that careful design of interactions between simulations is shown to drastically improve computational time.
The original idea of our work is that, at least for LQ problems, substantial gains can be made by coupling the simulations. For example, standard ARE solvers (in scipy) fail for large dimensions ($d = 200$) but an EnKF can still give a solution. EnKF is the preferred algorithm in data assimilation applications where the state-space is truly large.

Our approach is very relevant in the dynamical systems community where there many instances where one has a high fidelity simulator of the system, but obtaining and analytically solving an exact model is very hard, for instance, weather systems, fluid or aerodynamic simulations, and our approach is the first step to providing controllers for such systems using interacting particle systems, without the explicit use of models. 

There is however, a concern about robustness. If the real plant differs from simulations then one would need some online strategy to estimate and correct for such errors.

\subsection{Log transform}
The value function $\{v_t(x): 0 \leq t\leq
T, x\in\Re^d\}$ is defined as follows
\begin{align*}
v_t(x) & := \min_{U(\cdot) \in \mathbb{U}} \{ J_T(U) - J_t(U) \} \\
\text{s.t. } & \quad  \eqref{eq:dyn} \quad \text{ and } \quad X_0 = x.
\end{align*} 

Taking inspiration from literature for using log transform risk sensitive control \cite[Chapter 6]{fleming-2006-book}, we define a map $\bij : \Re
\to \Re$ so that 
\begin{align}\label{eq:pt}
p_t (x) : = \frac{\bij(v_t(x))}{\int \bij(v_t(x)) \ud x}, \quad 0\leq t\leq T,  \;x \in \Re^d
\end{align}
is a valid probability density function.  
The bijection $\bij$ is selected as
\begin{align}
\bij(z) \coloneqq 
\begin{cases}
\exp(-z); & \text{LQG} \\
\exp(-|\theta| z); & \text{LEQG}
\end{cases}.
\label{eq:bij}
\end{align}
Due to the quadratic nature of the value function, $p_t$ is the Gaussian density $\normal(0,S_t)$. The idea is to approximate $p_t$ using an ensemble of simulations, then obtain $P_t$ to find the optimal control. 

\subsection{Mean field system}

Define a stochastic process $\Ybarbar=\{\Ybarbar_t\in \Re^d:
0\leq t\leq T\}$ as a solution of the
following backward (in time) SDE:
\begin{subequations}\label{eq:Ybar}
\begin{align}
	\ud \Ybarbar_t & = A\Ybarbar_t\ud t  +  B\ud  \backward{\etabar}_t +  \sigma\ud  \backward{W}_t  +( \vf_t({Y}_t; \bar{n}_t, \bar{S}_t) + \correct_t({Y}_t; \bar{n}_t, \bar{S}_t)) \ud t      \\ 
	\Ybarbar_T & \sim \normal(0,S_T)
\end{align}
\end{subequations}
where ${\etabar}=\{{\etabar}_t \in \Re^m:0\leq t\leq T\}$ is a B.M. 
with a suitably chosen covariance matrix, $\vf_t(\cdot;\cdot), \correct_t(\cdot;\cdot)$ is a suitably chosen vector field, 
and $\bar{p}_t$ is the density of $\Ybarbar$.
\begin{proposition}
\label{prop:Y-exactness}
Consider the mean-field process~\eqref{eq:Ybar}.
Suppose $\cov(\eta)$,  $\vf$ and $\correct$ is selected according to Table \ref{tb:soln-lq}.    
Then, 
$
\bar{p}_t = p_t,\quad \forall t \in [0,T], 
$
where $\pbart$ is the probability density function of $\Ybarbar_t$ and $p_t$ is defined in~\eqref{eq:pt} in terms of the value function. The optimal control is expressed as a function of $\pbart$ according to
\begin{align*}
U_t^* = 
\begin{cases}
R\inv B\tp \grad\log\bar{p}_t(X_t); & \text{LQG} \\
(|\theta|R)\inv B\tp \grad\log\bar{p}_t(X_t); & \text{LEQG} \\
\end{cases}
\end{align*}
\end{proposition} 
\begin{proof}
See Section \ref{app:proof-mf}. 
\end{proof}

Note that the $\mathcal{A}_t$ in \eqref{eq:At} is the sum of $\vf$ and $\correct$. These quantities are expressed individually, since $\vf$ depends entirely on parameters obtained from the LQ cost, and $\correct$ depends entirely on parameters appearing in the coefficient for noise $\sigma$. 

\begin{table}
\caption{Vector fields for \eqref{eq:Ybar}.}
\vspace{1em}
\centering
\begin{tabular}{@{}cccc@{}} \toprule
& $\vf_t(z; n_t, S_t)$ & $\correct_t(z; n_t, S_t)$ & $\cov(\eta)$ \\ \midrule
LQG& $\half S_t C\tp C(z + n_t) $ & $\half \Sigma S_t\inv (z - n_t)$ & $R\inv$  \\ \midrule
LEQG& \multirowcell{4}{$\frac{|\theta|}{2} S_t C\tp C(z + n_t)$} &  \multirowcell{2}{$\Sigma S_t\inv (z - n_t)$} & \multirowcell{4}{$(\sqrt{|\theta|}R)\inv$} \\ 
$\theta > 0$ & &  \\ \cmidrule{1-1} \cmidrule{3-3}
LEQG&  &  \multirowcell{2}{$0$} &   \\ 
$\theta < 0$ & & &  \\ \bottomrule
\end{tabular}
\label{tb:soln-lq}
\end{table} 

\subsection{Finite-N approximation}

The mean-field process \eqref{eq:Ybar} is empirically approximated by simulating a
system of controlled interacting particles 
$\{Y^i_t \in\Re^d: 0 \leq t\leq T, i =1,\ldots,N\}$ 
according to \eqref{eq:dual_enkf_intro}.

\subsection{Details of mean field system}
\label{app:proof-mf}

What we need to show is that $\Ybarbar_t \sim \normal(0,S_t)$ for each $0 \le t \le T$.
Upon substitution of $\vf$ and $\correct$, the mean field system \eqref{eq:Ybar} for $\Ybarbar$ becomes, for LQG 
\begin{align*}
\ud\Ybarbar_t = A\Ybarbar_t\ud t + B\ud\backward{\eta}_t + \frac{1}{2}\bar{S}_tC\tp(C\Ybarbar_t + C\bar{n}_t)\ud t + \frac{1}{2}\sigma\sigma\tp\bar{S}_t\inv(\Ybarbar_t - \bar{n}_t)\ud t + \sigma \ud\backward{W}_t, \quad 
\end{align*}
and for LEQG 
for $\theta > 0$
\begin{align*}
\ud \Ybarbar_t &= A\Ybarbar_t\ud t + B \ud\backward{\eta}_t + \sigma \ud\backward{W}_t + \frac{\theta}{2}\bar{S}_tC\tp(C\Ybarbar_t + C\bar{n}_t)\ud t + \sigma\sigma\tp\bar{S}_t\inv(\Ybarbar_t - \bar{n}_t)\ud t, \nonumber 
\end{align*}
for $\theta < 0$
\begin{align*}
d\Ybarbar_t &= A\Ybarbar_t\ud t + B\ud\backward{\eta}_t + \sigma \ud\backward{W}_t - \frac{\theta}{2}\bar{S}_tC\tp(C\Ybarbar_t + C\bar{n}_t)\ud t,  \nonumber 
\end{align*}
where
\begin{align*}
\bar{n}_t = \E{\Ybarbar_t}, \, \bar{S}_t = \E{(\Ybarbar_t - \bar{n}_t)(\Ybarbar_t - \bar{n}_t)\tp}, \, \, \Ybarbar_0 \sim \normal(0,P_T^{-1}), \, \ud W_t \sim \normal(0,\identity \ud t)
\end{align*}
and $\eta$ is a Brownian motion with covariance as in Table \ref{tb:soln-lq}. 
Then we have for LQG
\begin{align*}
\dot{\bar{n}}_t &= (A + \bar{S}_tC\tp C)\bar{n}_t,  \\
\dot{\bar{S}}_t &= A\bar{S}_t + \bar{S}_tA\tp + \bar{S}_tC\tp C\bar{S}_t - BR^{-1}B\tp,
\end{align*}
and for LEQG
\begin{align*}
\dot{\bar{n}}_t &= (A + \frac{|\theta|}{2}\bar{S}_tC\tp C)\bar{n}_t  \\
\dot{{\bar{S}}}_t &= A\bar{S}_t + \bar{S}_tA\tp + |\theta|\bar{S}_t C\tp C\bar{S}_t - \frac{1}{|\theta|}(BR^{-1}B\tp - \theta\sigma\sigma\tp).
\end{align*}
The terminal conditions are for all cases, $\bar{n}_T = 0$ and $\bar{S}_T = S_T$. Since $\bar{n}_T$ is zero, then $\bar{n}_t = 0$ for all $0 \le t \le T$. And since $\bar{S}_t$ follows the same ODE as $S_t$ (in \eqref{eq:St}) and has the same terminal condition, it must be that $\bar{S}_t = S_t$ for all $0 \le t \le T$.

Finally, $Y$ is a Gaussian process since the SDE \eqref{eq:Ybar} is an Ornstein-Uhlenbeck process with a Gaussian terminal condition.

\section{Algorithms for implementation}
\label{app:alg}

Algorithm \ref{alg:P} is the offline algorithm to estimate the $\mathcal{Q}$ function, and Algorithm \ref{alg:EnKF} is the online algorithm to compute the optimal control.

\begin{algorithm}[t]
	\caption[Offline]{\textbf{[offline]} dual EnKF algorithm to approximate
      empirical Q function}
	\label{alg:P}
	\begin{algorithmic}[1]
		\REQUIRE Simulation time $T$, simulation step-size $\stepsize$, number of particles $N$, simulator $\simulator$ for \eqref{eq:dyn} (see Definition \ref{defn:sim}), terminal covariance $S_T$ from \eqref{eq:St}, running cost function $C$, and control cost matrix $R$, and risk parameter $\theta$ if applicable, the function $\mathcal{A}$ and covariance $\cov(\eta)$ from\eqref{eq:At} and \eqref{eq:cov}.
		\RETURN $\{P^{(N)}_k, \mathcal{Q}^{(N)}(\cdot,\cdot;k,\stepsize) : k=0,1,2,\ldots,\frac{T}{\stepsize} -1\}$	
		\STATE  $T_F = \frac{T}{\stepsize}$ 
		\STATE  Initialize $\{{Y}^i_{T_F}\}_{i=1}^N\iid \normal(0,S_T)$
		\STATE calculate $n^{(N)}_{T_F} = N^{-1}\sum_{i=1}^N { Y}^i_{T_F}$ 
		\STATE calculate $ S^{(N)}_{T_F}= (N-1)^{-1}\sum_{i=1}^N ({ Y}^i_{T_F} -n^{(N)}_{T_F})({ Y}^i_{T_F} - n_{T_F}^{(N)})\tp$
		\FOR{$k=T_F$  to $1$} 
\STATE Calculate	$ A^{(N)}_k= \mathcal{A}({Y}^i_t; n_t^{(N)}, S_t^{(N)}) $
		\FOR{$i=1$ to $N$}
		\STATE $\Delta \eta_k^i \iid \normal(0,\cov(\eta)\stepsize)$
		\STATE $\Delta{Y}^i_k = \simulator({ Y}^i_k,\Delta \eta_k^i,\stepsize ) +
		A^{(N)}_k  \stepsize$
		\STATE ${Y}^i_{k-1} = {Y}^i_k - \Delta{Y}^i_k$
		\ENDFOR	
		\STATE Calculate $n^{(N)}_{k-1} = N^{-1}\sum_{i=1}^N { Y}^i_{k-1}$ 
		\STATE Calculate	$ S^{(N)}_{k-1}= (N-1)^{-1}\sum_{i=1}^N ({ Y}^i_{k-1} -n^{(N)}_{k-1})({ Y}^i_{k-1} - n_{k-1}^{(N)})\tp$
		\STATE  Obtain $P^{(N)}_{k-1}$ from $\SN_{k-1}$ using \eqref{eq:St}, and $\mathcal{Q}^{(N)}(x,a;k,\stepsize)$ using \eqref{eq:empiricalQ}.
		\ENDFOR  
	\end{algorithmic}
\end{algorithm}

\begin{algorithm}[t]
    \caption{\textbf{[online]} dual EnKF algorithm to calculate optimal control}
    \label{alg:EnKF}
    \begin{algorithmic}[1]
    	\REQUIRE Simulation time $T$, simulation step-size $\stepsize$, number of averaging evaluations $N_e$,
    	empirical Q-function $ \mathcal{Q}^{(N)}(x,a,\stepsize)$ (see Definition \ref{defn:q}), $\{e_i\}_{i=1}^{m}$ the standard basis of $\R^m$.
    	\RETURN optimal control input  $\{  \hat{U}_k^{(N)} \in \R^m : k=0,1,2,\ldots,\frac{T}{\stepsize} -1  \}$.
   	\STATE Define $T_F\coloneqq \frac{T}{\stepsize}$ 	 
        \FOR{$k=0$ to $T_F-1$}
                	\STATE Observe state of the system, denoted $x_k$
					\STATE Define $y_k := P_k^{(N)} x_k, M_1 := 0$		
			\FOR{$j=1$ to $N_e$}      	
          	\STATE $M_1 \leftarrow  M_1 + \mathcal{Q}^{(N)}(x_k,0,\stepsize)$
          	\ENDFOR
          	\STATE $M_1 \leftarrow (N_e)^{-1} M_1$       	
          			\FOR{$i=1$ to $m$} 
          	\STATE Define $M_2 := 0$
          	
			\FOR{$j=1$ to $N_e$}      	
          	\STATE $M_2 \leftarrow  M_2 + \mathcal{Q}^{(N)}(x_k,R^{-1}e_i,\stepsize)$
          	\ENDFOR
          	\STATE $M_2 \leftarrow (N_e)^{-1} M_2$       
        	\STATE  $\ip{ \hat{U}^{(N)}_k}{ e_i } = M_2 - M_1 - \frac{1}{2}(R^{-1})_{ii}\stepsize$ 
        	\ENDFOR
          			\STATE Apply control $\hat{U}^{(N)}_k$ to the true system           			     	
        \ENDFOR
    \end{algorithmic}
\end{algorithm}

\section{Error Analysis}
\label{app:error}

\subsection{Obtaining bounds in \eqref{eq:error-S}}
\label{app:delmoral}

We get the bound \eqref{eq:error-S} from \cite[equation (2.10)]{delmoral-2019-matrix-ricc} (where the reader may also refer to Section 1.1, equation (1.4) and equation (3.7) of \cite{delmoral-2019-matrix-ricc} for more clarity). In the following, we go through the steps of obtaining the bounds \eqref{eq:error-S} using the aforementioned results from \cite{delmoral-2019-matrix-ricc}. The assumption $N \ge d + 1$ is justified in the end of \cite[Section 3.1]{delmoral-2019-matrix-ricc}.

We analyze the SDE \eqref{eq:Ybar}, which is the the mean field system for the particle system \eqref{eq:dual_enkf_intro}. We will analyze the system forward in time. To that end, consider the following mean field system for LQG 
\begin{align*}
\ud H_t = -AH_t\ud t + B\ud \eta_t - \frac{1}{2}\Omega_tC\tp(CH_t + Ch_t)\ud t - \frac{1}{2}\sigma\sigma\tp\Omega_t\inv(H_t - h_t)\ud t + \sigma \ud W_t, \quad 
\end{align*}
and for LEQG 
with $\theta > 0$
\begin{align*}
\ud H_t &= -AH_t\ud t + B\ud \eta_t + \sigma \ud W_t - \frac{\theta}{2}\Omega_tC\tp(CH_t + Ch_t)\ud t - \sigma\sigma\tp\Omega_t\inv(H_t - h_t)\ud t, \nonumber 
\end{align*}
for LEQG with $\theta < 0$
\begin{align*}
\ud H_t &= -AH_t\ud t + B\ud \eta_t + \sigma \ud W_t + \frac{\theta}{2}\Omega_tC\tp(CH_t + Ch_t)\ud t,  \nonumber 
\end{align*}
where
\begin{align*}
h_t = \E{H_t}, \, \Omega_t = \E{(H_t - h_t)(H_t - h_t)\tp}, \, \, H_0 \sim \normal(0,P_T^{-1}), \, \ud W_t \sim \normal(0,\identity \ud t)
\end{align*}
and $\eta$ is a Brownian motion with covariance as in Table \ref{tb:soln-lq}.
Then we have for LQG
\begin{align*}
\dot{h}_t &= -(A + \Omega_tC\tp C)h_t  \\
\dot{\Omega}_t &= -A\Omega_t - \Omega_tA\tp - \Omega_tC\tp C\Omega_t + BR^{-1}B\tp
\end{align*}
and for LEQG
\begin{align*}
\dot{h}_t &= -(A + \frac{|\theta|}{2}\Omega_tC\tp C)h_t  \\
\dot{\Omega}_t &= -A\Omega_t - \Omega_tA\tp - |\theta|\Omega_t C\tp C\Omega_t + \frac{1}{|\theta|}(BR^{-1}B\tp - \theta\sigma\sigma\tp)
\end{align*}


If the system is implemented using $N$ particles as follows for LQG,
\begin{align*}
\ud H_t^i &= -AH_t^i\ud t + B\ud \eta_t^i - \frac{1}{2}\Omega_t^{(N)}C\tp C(H_t^i + h_t^{(N)})\ud t \\ & \qquad - \frac{1}{2}\sigma\sigma\tp(\Omega_t^{(N)})\inv(H_t^i - h_t^{(N)})\ud t + \sigma \ud W_t^i 
\end{align*}
and for LEQG for $\theta > 0$
\begin{align*}
\ud H_t^i &= -AH_t^i\ud t + B\ud \eta^i_t + \sigma \ud W_t^i - \frac{\theta}{2}\Omega_t^{(N)}C\tp(CH_t^i \\ & \qquad + Ch_t^{(N)})\ud t - \sigma\sigma\tp(\Omega_t^{(N)})\inv(H_t^i - h_t^{(N)})\ud t 
\end{align*}
for $\theta < 0$
\begin{align*}
\ud H_t^i &= -AH_t^i\ud t + B\ud \eta_t^i + \sigma \ud W_t^i + \frac{\theta}{2}\Omega_t^{(N)}C\tp(CH_t^i + Ch_t^{(N)})\ud t 
\end{align*}
where
\begin{align*}
h_t^{(N)} = \frac{1}{N}\sum_{i=1}^{N}H_t^i, \quad \Omega_t^{(N)} = \frac{1}{N-1}\sum_{i=1}^{N}{(H^i_t - h^{(N)}_t)(H^i_t - h^{(N)}_t)\tp}
\end{align*}
and $\eta_t^i$ are iid copies of $\eta$.
 Then we have 
 The time-evolution for $\Omega^{(N)}_t$ is obtained by the application of the It\^o rule to its definition~\cite[Proposition 4.2]{delmoral-2023}
\begin{subequations}
\begin{align}
d\Omega_t^{(N)} &= (-A\Omega_t^{(N)} - \Omega_t^{(N)}A\tp - \Omega_t^{(N)}C\tp C\Omega_t^{(N)} + B_1R^{-1}_1B_1\tp)\ud t + \frac{1}{\sqrt{N}}\ud M_t \\
\ud M_t &= \frac{1}{\sqrt{N}}\sum_{i=1}^{N}(e^{i}_t(B\ud \eta_t^i + \sigma \ud W_t^i)\tp + (B\ud \eta_t^i + \sigma \ud W_t^i) (e^i_t)\tp), \quad e_t^i \coloneqq H_t^i - h_t^{(N)}
\end{align}
\label{eq:omegaN}
\end{subequations}

\begin{remark}\label{rmk:omegaN}
Observe that  $\Omega_t={S}_{T-t}$ and $\Omega^{(N)}_t =S^{(N)}_{T-t}$.
\end{remark}

Finally, we get the bound \eqref{eq:error-S} from \cite[equation (2.10)]{delmoral-2019-matrix-ricc}. To see how to apply the result to our case, notice that $\Omega_t^{(N)}$ follows the same dynamics as \cite[equation (3.7)]{delmoral-2019-matrix-ricc} (where the reader may refer Section 1.1 (with particular emphasis on equation (1.4)) of \cite{delmoral-2019-matrix-ricc} for better clarity).

\subsection{Obtaining bounds in \eqref{eq:error-Sinf}}
\label{app:aajerror}

These proofs are largely based on the proofs in \cite{anant-2022}.

\renewcommand{\OmegaN}{\Omega^{(N)}}
\newcommand{\BB}{\Sigma_B}
{\noindent \bf Notation:} Let $S^d_{+} \subset S^d \subset \R^{d \times d}$ denote the set of symmetric positive definite matrices and symmetric matrices respectively. Let $\langle Q_1, Q_2\rangle \coloneqq \trace(Q_1Q_2^\top)$ denote the Frobenius inner product for $Q_1,Q_2 \in \R^{d \times d}$.  Then $||\cdot||_{F} \coloneqq \sqrt{\langle Q_1, Q_1 \rangle}$.

In this section (that is, in Appendix \ref{app:aajerror}), for LEQG, we redefine $C \leftarrow \sqrt{|\theta|}C$ to keep notation the same for all problems for further analysis.

From Appendix \ref{app:delmoral} we know that $\Omega_t$ satisfies the Riccati equation
\begin{equation}\label{eq:Sbar-t}
\dot{\Omega}_t =  \Ricc(\Omega_t):=-A\Omega_t  - \Omega_t A^\transpose  - \Omega_tC^\transpose C \Omega_t + \BB,
\end{equation}  
where $\BB:=BR^{-1}B^\top$ for LQG and $\BB:=|\theta|\inv(BR^{-1}B^\top - \theta\Sigma)$ for LEQG with $\Sigma := \sigma\sigma\tp$. From \eqref{eq:omegaN} we know that 

	\begin{equation}\label{eq:SNt}
	\ud \OmegaN_t = \Ricc( \OmegaN_t)\ud t + \frac{1}{\sqrt{N}}\ud M_t,
	\end{equation} 
	where $\{ M_t : t \ge 0 \}$ is a martingale 
	given by
	\begin{align*}
	\ud M_t &= \frac{1}{N-1}\sum_{i=1}^{N}(e^{i}_t(B\ud \eta_t^i + \sigma \ud W_t^i)\tp + (B\ud \eta_t^i + \sigma \ud W_t^i) (e^i_t)\tp), \quad e_t^i \coloneqq H_t^i - h_t^{(N)}
	\end{align*}
	with quadratic variation
		\begin{align*}
 \ud\langle M\rangle_t = (\trace(\BB) + \Sigma) \OmegaN_t + (\BB + \Sigma) \trace( \OmegaN_t) + (\BB + \Sigma)  \OmegaN_t  +\OmegaN_t  (\BB + \Sigma) 
	\end{align*}

Let $\phi(t,Q)$ denote the semigroup associated with the Riccati equation such that for any positive definite matrix $Q \in S^d_+$, 
\begin{equation*}
	\frac{\partial \phi}{\partial t} (t,Q) = \Ricc(\phi(t,Q)),\quad \phi(0,Q)=Q. 
\end{equation*}
We define the first-order and the second-order derivatives which are the linear and bilinear operators $\frac{\partial \phi}{\partial Q}(t,Q):S^d \to S^d$ and $\frac{\partial^2 \phi}{\partial Q^2}(t,Q):S^d\times S^d\to S^d$ respectively as
\begin{align*}
\frac{\partial \phi}{\partial Q}(t,Q) (Q_1)   &:= \left.\frac{\ud}{\ud \epsilon}\right\vert_{\epsilon=0} \phi(t,Q+\epsilon Q_1) \\
\frac{\partial^2 \phi}{\partial Q^2}(t,Q)(Q_1, Q_1)  &:= \left.\frac{\ud^2}{\ud \epsilon^2} \right\vert_{\epsilon=0}\phi(t,Q+\epsilon Q_1) .
\end{align*}
We denote by $\| \frac{\partial \phi}{\partial Q}(t,Q) \|_{F,F}$ and $\| \frac{\partial^2 \phi}{\partial Q^2}(t,Q)\|_{F,F}$ induced-norm of these operators with respect to the Frobenius norm.   The following lemma is an intermediate result. 
\begin{lemma}
	For $\Omega_t$ and $\OmegaN_t$ defined in~\eqref{eq:Sbar-t} and \eqref{eq:SNt} respectively, the following is true,
	\begin{equation}
	\begin{aligned}
		\OmegaN_t-\Omega_t&= \frac{1}{\sqrt{N}}\int_0^t \frac{\partial \phi}{\partial Q}(t-s,\OmegaN_s) (\ud M_s) \\&+ \frac{1}{2N}\int_0^t \frac{\partial^2 \phi}{\partial Q^2}(t-s,\OmegaN_s)(\ud M_s, \ud M_s) + \phi(t,\OmegaN_0) - \phi(t,\Omega_0) 
	\end{aligned}
\label{eq:integral-representation}
\end{equation}
\end{lemma}

\begin{proof}
	We see that 
\begin{align*}
	\Omega^{(N)}_t -& \Omega_t =  \phi(0,\Omega^{(N)}_t)  - \phi(t,\Omega_0)\\
	&=   \phi(0,\Omega^{(N)}_t)  - \phi(t,\Omega^{(N)}_0)  +\phi(t,\Omega^{(N)}_0) - \phi(0,\Omega_0)
	\\&= \int_0^t\ud_s\phi(t-s,\Omega^{(N)}_s) +  \phi(t,\Omega^{(N)}_0) - \phi(t,\Omega_0).
\end{align*}
Evaluating the differential we have,
\begin{align*}
	\ud_s \phi(t-s,\Omega^{(N)}_s) &=   -\frac{\partial \phi}{\partial t} 
	(t-s,\Omega^{(N)}_s) \ud s + \frac{\partial \phi}{\partial Q}(t-s,\Omega^{(N)}_s) (\ud \Omega^{(N)}_s)  \\& \quad +  \frac{1}{2}\frac{\partial^2 \phi}{\partial Q^2}(t-s,\Omega^{(N)}_s) (\ud \Omega^{(N)}_s,\ud \Omega^{(N)}_s), 
\end{align*}
where we used the identity $ \frac{\partial \phi}{\partial t}(t,Q) =  \frac{\partial \phi}{\partial Q}(t,Q) (\text{Ricc}(Q))$. 
\end{proof}


We need the following assumption to use the aforementioned lemma to arrive at the desired result.
\begin{assumption}
	Consider the  semigroup corresponding to the Riccati equation~\eqref{eq:Sbar-t}. There are positive constants $c_1$, $c_2$, and $\lambda$ such that $\forall Q \in S^d_+$: 
	\begin{align*}
		\|\frac{\partial \phi}{\partial Q}(t,Q)\|_{F,F} \leq c_1e^{-2\lambda t},\quad 	\|\frac{\partial^2 \phi}{\partial Q^2}(t,Q)\|_{F,F} \leq c_2e^{-2\lambda t}. 
	\end{align*}
\label{assumption:Riccati}
\end{assumption}
Exponential decay holds for $(A,B)$ controllable and $(A,C)$ observable,  \cite[Section 2]{delmoral-2023}. However, the for the constants $c_1$ and $c_2$ to be the same for initial $Q$, see~\cite[Section 4.2]{bishop-2017} for detailed analysis of the Riccati equation under the additional assumption that  the matrix $C$ is full-rank.     

\begin{proposition}\label{prop:error-analysis-appendix}
	If Assumption~\ref{assumption:Riccati} holds, the the following upper-bound (repeated from \eqref{eq:error-Sinf}) is true
\begin{equation*}
\Expect[\|S^{(N)}_{t}-{S}_t\|_F] \leq
\frac{C_5}{\sqrt{N}} + C_6e^{-2\lambda (T-t)} \Expect[\|\SN_T-S_T\|_F],
\end{equation*}
where $C_5,C_6$ are  time-independent positive constants.
\end{proposition}
\begin{proof}
	Using triangle inequality for norm on ~\eqref{eq:integral-representation} we get
	\begin{align*}
		\Expect[\|\OmegaN_t-\Omega_t\|_F] \leq \frac{r_1}{\sqrt{N}} + \frac{r_2}{2N} + r_3
	\end{align*}
where  we define
\begin{align*}
	r_1 & := \Expect \left[\left\|\int_0^t \frac{\partial \phi}{\partial Q}(t-s,\Omega_s)(\ud M_s)\right\|_F\right]\\ 
	r_2 & := \Expect  \left[\int_0^t \left\|\frac{\partial^2 \phi}{\partial Q^2}(t-s,\Omega_s)(\ud M_s,\ud M_s)\right\|_F\right]  \\ 
	r_3 & := \Expect \left[\left\| \phi(t,\OmegaN_0) - \phi(t,\Omega_0)\right \|_F\right]
\end{align*}
Now we get bounds for $r_1,r_2$ and $r_3$. For $r_1$,
\begin{align*}
	r_1 
	&\leq  \left[\Expect\left[ \left\|\int_0^t \frac{\partial \phi}{\partial Q}(t-s,\Omega_s)(\ud M_s)\right\|_F^2 \right]\right]^\half\\
	&=\left[\int_0^t \Expect \left[\left\| \frac{\partial \phi}{\partial Q}(t-s,\Omega_s)(\ud M_s)\right\|_F^2\right] \right]^\half\\
	& \leq \left[\int_0^t \Expect \left[\| \frac{\partial \phi}{\partial Q}(t-s,\Omega_s)\|_{F,F}^2 \|\ud M_s\|_F^2\right] \right]^\half   \\
	&\leq \left[\int_0^t 4c_1^2e^{-4\lambda (t-s)}\trace(\Sigma_B + \Sigma)\Expect[\trace(\OmegaN_s)]\ud s\right]^\half
\end{align*}
where in the first inequality we used Jensen's inequality, for the second inequality we used It\"o isometry in the second step, and Assumption~\ref{assumption:Riccati} in the last inequality. For $r_2$ we used Assumption~\ref{assumption:Riccati} to see that 
\begin{align*}
	r_2 &\leq \Expect  \left[\int_0^t \|\frac{\partial^2 \phi}{\partial Q^2}(t-s,\Omega_s) \|_F \| \ud M_s\|^2_F\right]\\
	&\leq  \int_0^t 4c_2e^{-2\lambda (t-s)}\trace(\Sigma_B + \Sigma) \Expect[\trace(\OmegaN_s) ]\ud s
\end{align*}
	For the bounds on $r_3$ we use the bounds on the first derivative in Assumption~\ref{assumption:Riccati} to get
	\begin{align*}
	r_3 \leq c_1e^{-2\lambda t}\Expect[\|\OmegaN_0 - \Omega_0\|_F]
	\end{align*} 
Upon using the bound  $\Expect[\trace(\OmegaN_t)] \leq \trace(\Omega_t) $ from~\cite[Theorem 5.2]{delmoral-2023}, and from exponential convergence of $\Omega_t$ to $\bar{\Omega}$, there exists $E_0 \in (0,\infty)$ such that $\trace(\Omega_t) \leq \sup_{t\geq  0} \trace(\Omega_t)\le E_0$ we get that
\begin{align*}
		\Expect[\|\OmegaN_t-\Omega_t\|_F] \!\leq \!(c_1 \!+ \!c_2\sqrt{\epsilon})\sqrt{\epsilon} \!+ \!c_1e^{-2\lambda t} \Expect[\|\OmegaN_0 - \Omega_0\|_F]
	\end{align*}
where $\epsilon: = \frac{E_0\trace(\Sigma_B + \Sigma)}{\lambda N}$. Making a change of variable from $t$ to $T-t$ and recalling Remark \ref{rmk:omegaN} concludes the proof.  
\end{proof}

\subsection{Obtaining bounds in \eqref{eq:error-U}}
\label{app:errorU}

Fix a time $t \in [0,T]$ and $x \in \R^d$. Fix a step size $\stepsize$. 
Recall the Q function is, with $y := P_t^{(N)}x$,
\begin{align*}
\mathcal{Q}(x,a,\stepsize) &= y^T\simulator(x,a,\stepsize) +  (\frac{1}{2}|Cx|^2+\frac{1}{2}a^\top Ra)\stepsize \\
&=  y^T((Ax + Ba)\stepsize + \sigma\Delta W) +  (\frac{1}{2}|Cx|^2+\frac{1}{2}a\tp R a)\stepsize
\end{align*}
with $\Delta W \sim \normal(0,\identity \stepsize)$. To obtain the optimal control, we minimize the Hamiltonian with respect to $a$ after substituting the value of the momentum $y$ in terms of $x$.  
Define $K_t := -R\inv B\tp P_t$, $U_t^{\text{opt}} := K_tx$, $K_t^{(N)} := -R\inv B\tp P_t$ and $U_t^{(N)} \coloneqq K_t^{(N)}x$.
%
The first step is getting an expression for $\hat{U}_t^{(N)}$.
The Q function is, 
\begin{align*}
\mathcal{Q}(x,a,\stepsize) &= \frac{1}{2}a^\top Ra \stepsize + (x\tp P_t^{(N)} B \stepsize)a +  \alpha\tp\Delta W + \varphi(x)\\
\varphi(x) &\coloneqq \frac{1}{2} x\tp (Q\stepsize + P_t^{(N)}A\stepsize + A\tp P_t^{(N)}\stepsize)x, \quad \alpha \coloneqq \sigma\tp P_t^{(N)} x 
\end{align*}
To be consistent with notation used in the algorithm, define
\begin{align*}
M_1 \coloneqq & \frac{1}{N_e} \sum_{i=1}^{N_e} \mathcal{Q}(x,0,\stepsize) = \varphi(x) + \alpha\tp (\Delta W)_1 \\
(\Delta W)_1 &:= \frac{1}{N_e}\sum_{i=1}^{N_e}(\Delta W) \sim \normal(0,\frac{\stepsize}{N_e})\\
\quad M_2^i \coloneqq & \frac{1}{N_e} \sum_{i=1}^{N_e} \mathcal{Q}(x,R\inv e_i,\stepsize)
\end{align*}
where the summation denotes that each call of the Hamiltonian function produces an independent realization of the random variable $\Delta W$. 
Now,
\begin{align*}
\mathcal{Q}(x,R\inv e_i,\stepsize) = \half (R\inv)_{ii}\stepsize + \beta_i\stepsize + \alpha\tp \Delta W + \varphi(x)
\end{align*}
where $\beta_i := e_i\tp(R\inv B\tp P_t^{(N)}x) = \ip{U_t^{(N)}}{e_i}$, which gives
\begin{align*}
\mathcal{Q}(x,R\inv e_i,\stepsize) - \half (R\inv)_{ii}\stepsize - M_1 &= \beta_i\stepsize + \alpha\tp \Delta W + (\varphi(x) - M_1) \\
&= \beta_i\stepsize + \alpha\tp ( \Delta W - (\Delta W)_1 )
\end{align*}
Therefore, 
\begin{align*}
M_2^i - \half (R\inv)_{ii}\stepsize - M_1 = \beta_i\stepsize + \alpha\tp ( (\Delta W)_2 - (\Delta W)_1 ).
\end{align*}
where again $(\Delta W)_2 := \frac{1}{N_e}\sum_{i=1}^{N_e}(\Delta W) \sim \normal(0,\frac{\stepsize}{N_e})$ and the summation similarly denotes that each function call of the Hamiltonian gives an independent realization of $\Delta W$.
Since by definition, $\ip{\hat{U}^{(N)}_t}{e_i} = (M_2^i - \half (R\inv)_{ii}\stepsize - M_1)\frac{1}{\stepsize}$ we have
\begin{align}\label{eq:uhatn}
\ip{\hat{U}^{(N)}_t}{e_i} = \ip{{U}^{(N)}_t}{e_i} + \alpha\tp \left(\frac{ (\Delta W)_2 - (\Delta W)_1}{\stepsize} \right).
\end{align}
Define 
\begin{align*}
\omega := P_t^{(N)}\sigma \xi, \quad \xi := \left( \frac{ (\Delta W)_2 - (\Delta W)_1}{\stepsize}\right) \sim \normal(0,\frac{2}{N_e\stepsize}).
\end{align*}
Then from \eqref{eq:uhatn} we see that $\hat{U}^{(N)}_t = {U}^{(N)}_t + ({\ones}\omega\tp)x $. Thus, we define $\hat{K}_t^{(N)} := K_t^{(N)} + {\ones}\omega\tp$, where $\ones$ denotes the vector with each entry equal to 1, to get $\hat{U}^{(N)}_t = \hat{K}_t^{(N)}x $.
Now we give the mean square error between  $\hat{K}_t^{(N)}$ and $K_t$ as
\begin{align*}
\half\E{\| \hat{K}_t^{(N)} - K_t \|^2} \le \E{\| \hat{K}_t^{(N)} - K_t^{(N)} \|^2} + \E{\| {K}_t^{(N)} - K_t^* \|^2}
\end{align*}
The first term can be estimated as
\begin{align*}
\E{\| {K}_t^{(N)} - K_t \|^2} &= \E{\| \ones\tp \omega \|^2} =  n\E{\| P_t^{(N)}\sigma \xi\|^2}\\
&\le 2n\E{\|P_t^{(N)}-P_t\|^2 \|\sigma\|^2 |\xi|^2} + 2n\|P_t\|^2 \|\sigma\|^2\E{ |\xi|^2} \\
&= \frac{2n}{N_e\stepsize}(\frac{\tilde{C}_1}{N} + \tilde{C}_2)
\end{align*}
where we used \eqref{eq:error-S} and the fact that $\xi$ and $P_t^{(N)}$ are independent random variables, and exponential convergence of $P_t$ to $P_{\infty}$ ensures a uniform bound on $\|P_t\|$.
The second term can be estimated as
\begin{align*}
\E{\| {K}_t^{(N)} - K_t \|^2} &= \E{\| R\inv B\tp ({P}_t^{(N)} - P_t) \|^2} \\
& \le \|R\inv B\tp\|^2 \frac{\tilde{C}_3}{N}
\end{align*}
using properties of matrix norms and equivalence of $\|\cdot\|$ and $\|\cdot\|_F$. 

\section{Simulation details for numerical comparisons}
\label{app:numcomp}

\subsection{Policy optimization}
\label{app:numcomp-po}

We compare our algorithm with [K19] and [Z21]. Codes for [K19] were found in the supplementary material of their paper \cite{krauth-2019}, while codes for [Z21] are on github \cite{kzcodes} as a part of the paper \cite{kz-codes-paper} which builds on \cite{zhang-2021-neurips}. 

\subsubsection{Discussion of results}
\label{app:compdisc}

\textbf{Comparison with [K19]:} from the sample complexity comparison in Table \ref{tb:samp-comp}, we see that both [K19] and dual-EnKF have similar sample complexity. However, since [K19] is a policy gradient type algorithm, they need to run copies of the LQG system forward in time for each iteration of their algorithm. Since we need to execute only one iteration of the linear dynamical system, it can be expressed as vector matrix multiplications in python, and we use that structure to leverage the vectorization capabilities of numpy to obtain an order of magnitude acceleration in simulation time. 

\textbf{Comparison with [Z21]:} from the sample complexity comparison in Table \ref{tb:samp-comp}, we see that [Z21] has a much higher sample complexity than dual EnKF. it stems from that fact that [Z21] have a policy gradient type approach, which simulates the system forward in time for each iteration. Moreover, [Z21] estimates the finite horizon gain as a function of time, which requires stacking all the gains into one large matrix, that increases the problem size significantly.

\subsubsection{Model and simulation parameters}
We run all three algorithms on a discrete time  used in [Z21], and plot the simulation time required to reach a specified relative error in gain and cost. We recall that [Z21] considers a finite time LEQG problem with $\theta > 0$ and [K19] considers an infinite horizon LQG problem.  Both works are in discrete time, and the details of the dynamical system, the optimal control parameters, and simulation parameters are all below. We convert the discrete time system to a continuous time system for running the dual EnKF (conversion formulas in Appendix \ref{app:numcomp}).

The discrete time system has the following parameters (same as the one in \cite[Section 5]{zhang-2021-neurips}):
\begin{equation*}
	\begin{aligned}
		A_d &= \begin{bmatrix}
			1 & 0 & -5\\
			-1 & 1 & 0\\
			0 & 0 & 1
		\end{bmatrix},\quad
		B_d  = \begin{bmatrix}
			1 & -10 & 0\\
			0 & 3 & 1 \\
			-1 & 0 & 2 
		\end{bmatrix},\quad
		\sigma_d = \begin{bmatrix}
			5 & 0 & 0\\
			0 & 2 & 0\\
			0 & 0 & 2
		\end{bmatrix}\\
		Q_d &= \begin{bmatrix}
			1 & 0 & 0\\
			0 & 1 & 0\\
			0 & 0 & 1
		\end{bmatrix},~\text{and}~
		R_d = \begin{bmatrix}
			4 & -1 & 0\\
			-1 & 4 & -2\\
			0 & -2 & 3
		\end{bmatrix}
	\end{aligned}
\end{equation*}
We convert it to continuous time using a first order approximation with a discretization step size $\stepsize_d = 0.1$s  as follows:
\begin{align*}
A &= \frac{\log(A_d)}{\stepsize_d},\quad 
B = \frac{B_d}{\stepsize_d}, \quad
\sigma = \frac{\sigma_d}{\stepsize_d} \\
Q &= \frac{Q_d}{\stepsize_d}, \quad  R = \frac{R_d}{\stepsize_d} , \quad G = Q_d 
\end{align*}
A simulation step size $\stepsize = 0.002$s is used in our dual EnKF algorithm and for the finite time LEQG simulation, the risk parameter is set as $\theta = 0.2$ with the simulation time horizon as $T=0.5$s, i.e. 5 discrete time steps.

\subsubsection{Description of procedure}

In the Figure \ref{fig:comparison}, the results of the number of particles with $N=100, 200, 300, \ldots, 1000$ in our dual EnKF algorithm are presented. 
For the algorithm in [K19], iteration steps of 100,000 is used, and the results of time horizon as 3000, 3500, 4000, 5000, 6000 and 6500 time steps are presented.
For the algorithm in [Z21], six linear spacing results of error in cost (and in gain), ranging from min error with 100,000 iteration steps to 10\% (and 60\%), are presented.
All results are averaged over 100 runs to find the expectation and the standard deviation. All the results plotted in Figure \ref{fig:comparison}. The [Z21] and [K19] algorithms are run till a certain error in cost or gain is reached, and the computation time is recorded. The EnKF is always run for 10 second simulation period, and the computation time is recorded. Therefore, we document the computation time required to reach a certain value of error. Since the PO algorithms have random number generators, the computation time needed to achieve a certain error lies in a range of values. The EnKF also has a random number generator hence it also exhibits a range of values, but they are very tightly concentrated around the mean.

All the simulations are executed on a desktop iMac computer equipped with a 3 GHz 6-Core Intel Core i5 processor with python3.
The device also has a 32GB 2667MHz DDR4 memory and a Radeon Pro 560X 4GB graphic card. Simulation times were measured in python using the {\tt time.time()} function found in the {\tt time} module.

When analyzing error in gain ($\epsilon^{\text{gain}}$), we first recall that for [Z21] and dual EnKF, due to the finite time horizon, both $ K^{\text{alg}}$ and $\bar{K}^{\text{opt}}$ are functions of time, while for [K19] it is only one values, the infinite horizon gain. To compute the error in gain, we first need to find the optimal gain ($\bar{K}^{\text{opt}}$) and the gain output by the algorithm ($ K^{\text{alg}}$). For [Z21] and [K19] we use the output directly from the codes provided in the following manner. Both codes solve the Riccati equation and output $\bar{K}^{\text{opt}}$. To find the error, we calculate the relative error between the $ K^{\text{alg}}$ obtained after each iteration, and $\bar{K}^{\text{opt}}$. For dual EnKF, we find $\bar{K}^{\text{opt}}$ by solving the Riccati equation and compare it with the $\bar{K}^{\text{alg}}$ from the algorithm. To get an estimate of the infinite horizon gain using the dual EnKF, we simply consider refer to the formula in Section \label{sec:optcon}.

To calculate the error in cost ($ \epsilon^{\text{cost}}$), we need the cost produced by the system when the gain produced by the algorithm is applied to it ($c^{\text{alg}}$), and the cost produced on application of the optimal gain ($c^{\text{opt}}$). For [Z21] and [K19] we use the output directly from the codes provided. The codes calculate both $c^{\text{alg}}$ and $c^{\text{opt}}$. For dual EnKF, to find $c^{\text{opt}}$ and $c^{\text{alg}}$, we find the cost incurred by applying the optimal infinite horizon gain and estimated infinite horizon gain respectively to the system. Given a gain, the cost incurred is computed by solving a Lyapunov equation \cite{davis} for LQG or running a system forward in time and averaging the cost incurred for LEQG.

The simulation time for EnKF records time needed to execute both Algorithm \ref{alg:P} and \ref{alg:EnKF}.

\subsection{Path integral control}
\label{app:numcomp-pi}

\subsubsection{Discussion of results}

We observe that dual EnKF provides over an order of magnitude gain in simulation time, even though the path integral approach is a model based approach. In this simulation as well, we leverage python computation speed in performing matrix vector multiplications.

\subsubsection{Model and simulation parameters}

We run both algorithms on the spring mass damper system (Appendix \ref{app:smd}). We discretize the system when implementing the MPPI algorithm using a discretization step size of 0.1s. We vary the number of masses in the system to vary the dimension of the state.

\subsubsection{Description of procedure}



In Figure \ref{fig:comparison}, we present a comparative analysis of two algorithms: the dual EnKF algorithm and the MPPI algorithm. Both algorithms were evaluated over a 10-second simulation period, with each configuration tested across 15 independent runs to determine average performance metrics and their standard deviations. Therefore, we document the error achieved and the corresponding computation time achieved by that algorithm. Since both algorithms have random number generators, for the same computation time, there is a range of errors that are achieved.

For the dual EnKF algorithm, we tested four different particle configurations ($N =$ 100, 500, 1000, and 5000 particles), while the MPPI algorithm was evaluated using four different particle counts (10, 50, 100, and 500 particles).

To assess scalability, we tested both algorithms across three system dimensions corresponding to 5 masses (10-dimensional system), 10 masses (20-dimensional system), and 20 masses (40-dimensional system). For each configuration, we measured two key performance metrics: computation time and cost error. This experimental setup allows us to analyze how both algorithms perform across different particle counts and system dimensions.

All the simulations are executed on a desktop iMac computer equipped with a 3 GHz 6-Core Intel Core i5 processor with python3.
The device also has a 32GB 2667MHz DDR4 memory and a Radeon Pro 560X 4GB graphic card. Simulation times were measured in python using the {\tt time.time()} function found in the {\tt time} module.

The simulation time for EnKF records time needed to execute both Algorithm \ref{alg:P} and \ref{alg:EnKF}.

\section{Numerical illustration of error formulas~\eqref{eq:error-S},\eqref{eq:error-Sinf})}

\subsection{Spring mass damper model}
\label{app:smd}
\newcommand{\dsp}{d_s}
\newcommand{\toep}{\mathbb{T}}
\newcommand{\id}{\identity}

This system is taken from~\cite{mohammadi_global_2019}.
Let the number of masses be $d_s$.
The matrices $A$ and $B$ are as follows:
\begin{align*}
A = \begin{bmatrix}
0_{\dsp \times \dsp} & \id_{\dsp} \\ -\toep & -\toep
\end{bmatrix}, \quad B = \begin{bmatrix}
0_{\dsp \times \dsp} \\ \id_{\dsp}
\end{bmatrix}
\end{align*}
then the dimension of the system is $d = 2d_s$, and  $\toep \in
 \R^{\dsp \times \dsp}$ is a  Toeplitz matrix with $2$ on the main
 diagonal and $-1$ on the first sub-diagonal and first
 super-diagonal. We let $C,R,G$ be identity matrix of suitable dimension. The two values of $\theta$ are $\{-0.8,1.1\}$. For Figure \ref{fig:mse} $\sigma = 0.1 B$, $T = 10$s and we average MSE data over 500 runs to find the expectation, and for Figure \ref{fig:smdenergy} $\sigma = 0.3B$, $T = 5$s and we average energy data over 100 runs to find the expectation. The simulation step size is $\tau = $0.02s for both. For Figure \ref{fig:smdenergy} we use 500 particles for all simulations.

We evaluate the control algorithm obtained from the dual EnKF on the spring mass damper system and plot the energy of the system (defined as the norm square of the state). We use 1000 particles, and results are shown over an average of 100 simulations. We see that for as high as 80 dimensions, our algorithm manages to reduce the energy and keep it sufficiently close to zero. Results are found in Figure \ref{fig:smdenergy} for the stable and unstable spring mass damper system (where the latter is a mathematical construction obtained by reversing the sign of $A$ to change stability properties of the uncontrolled system).

%

\begin{figure}
\centering
\includegraphics[scale=0.4]{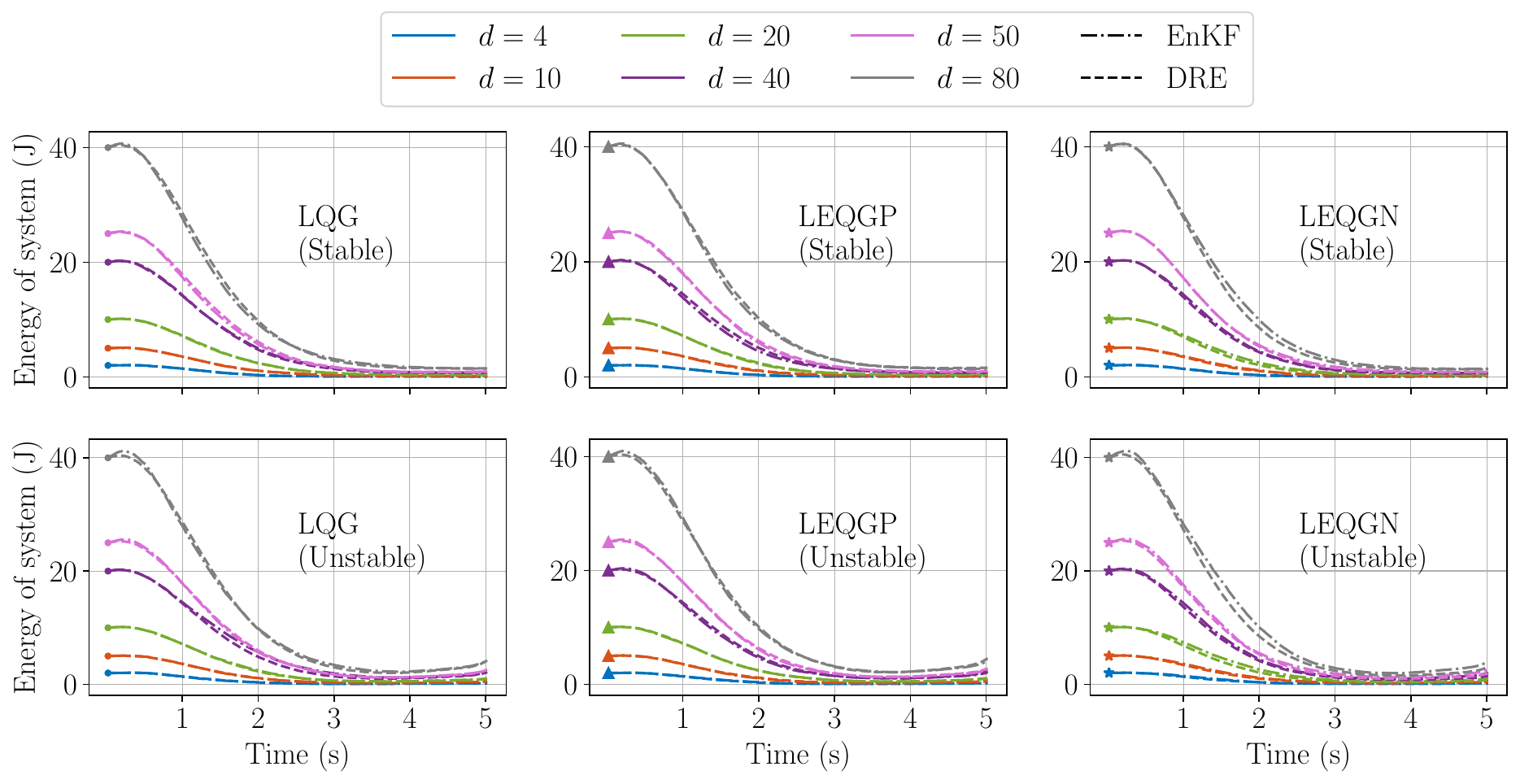}
\caption{Performance of all three controllers on stable spring mass damper system.}
\label{fig:smdenergy}
\end{figure}

\subsection{Random System}
\label{app:addnumrandom}

We choose a random d-dimensional system is in its controllable canonical form with
\begin{equation*}
A = \begin{bmatrix}
0&1&0&0&\ldots&0\\
0&0&1&0&\ldots&0\\
\vdots & & & & & \vdots\\
a_1 & a_2 & a_3 & a_4 &\ldots & a_d
\end{bmatrix},\quad 
B = \begin{bmatrix}
0\\
0\\
\vdots \\
1 
\end{bmatrix}
\end{equation*}
where $(a_1,\ldots,a_d)\in \mathbb{R}^d$ are i.i.d.
samples drawn from $\normal(0 ,1)$.  The matrices $C,R,G$, are identity
matrices of appropriate dimension $\sigma = 0.1B$ and $\theta \in \{1.1,-0.8\}$ .  For all simulations, $T=10$, and $\stepsize = 0.02$, and $N=500$ particles.

Figure~\ref{fig:convergence} shows the convergence of the  100 entries in $P_t^{(N)}$ to the solution of the ARE.
Figure \ref{fig:poles} shows the
open-loop poles (eigenvalues of the matrix $A$) and the closed-loop
poles (eigenvalues of the matrix $(A + BK_0^{(N)})$). As noted earlier, the closed-loop poles are all stable, whereas some
open-loop poles have positive real parts.      


\begin{figure}[h]
		\centering
{	
	\subfigure{
         \includegraphics[scale=0.3]{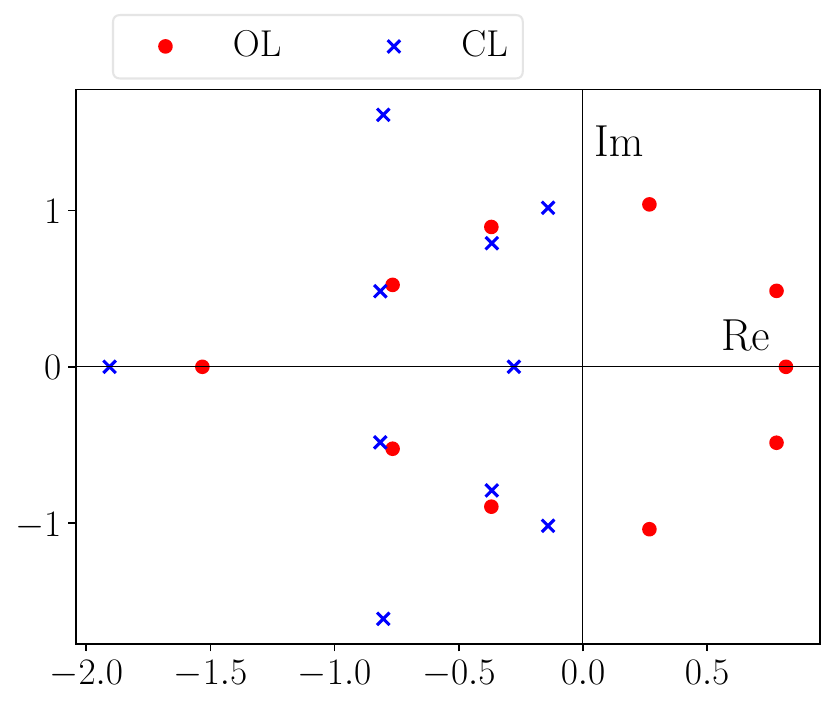}  
         \label{fig:lqg-poles}
    }
	\subfigure{
         \includegraphics[scale=0.3]{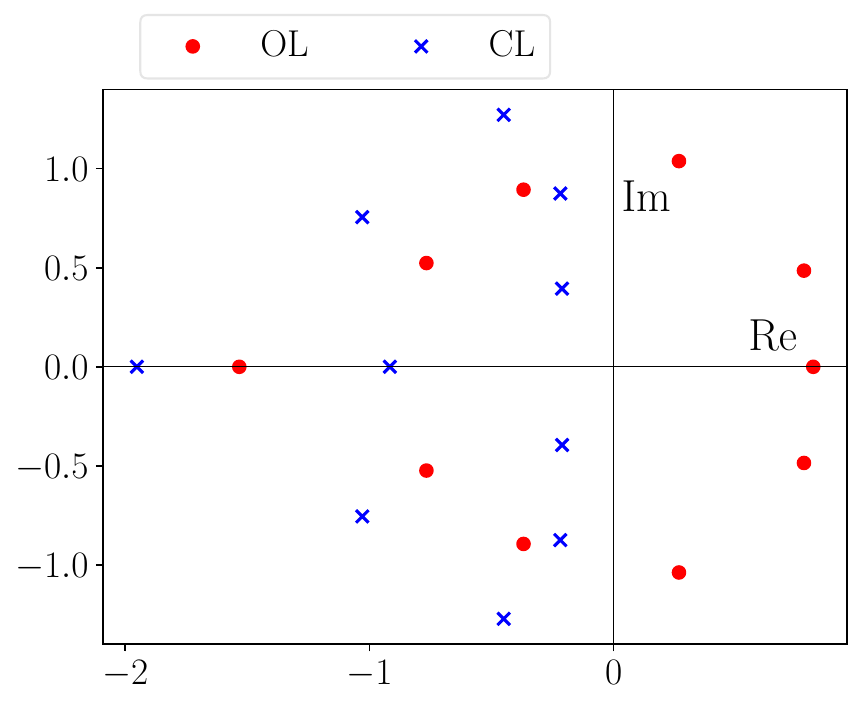}  
         \label{fig:leqgp-poles}
     }
     \subfigure{
         \includegraphics[scale=0.3]{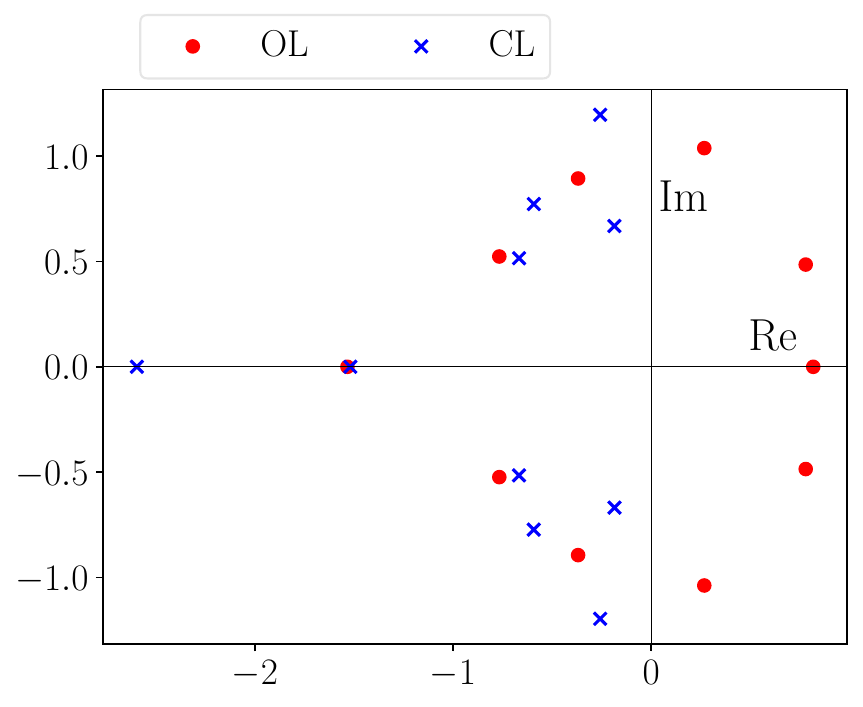}  
         \label{fig:leqgn-poles}
     }
}
		\caption{Open and closed-loop poles for (a) LQG (b) LQEG ($\theta > 0$) (c) LEQG ($\theta < 0$).}
		\label{fig:poles}
		\vspace{-0.1in}
	\end{figure}

\end{document}